\documentclass[10pt, two column, twoside]{IEEEtran}

\usepackage{amssymb}
\usepackage{amsmath}
\usepackage[lined,boxed,commentsnumbered, ruled]{algorithm2e}
\usepackage{mathrsfs}
\usepackage{algorithmic}
\usepackage{bm}
\usepackage{tikz}
\usetikzlibrary{arrows}
\usepackage{subfig}
\usepackage{graphicx,booktabs,multirow}

\definecolor{color1}{HTML}{D0B22B}
\definecolor{dred}{RGB}{128,0,0}
\definecolor{colorhkust}{RGB}{20,43,140}
\definecolor{colorshanghaitech}{RGB}{162,0,5}
\definecolor{colortsinghua}{RGB}{116,52,129}
\definecolor{colordark}{RGB}{184,134,11}
\usepackage{amsthm}
\theoremstyle{definition}
\newtheorem{lemma}{Lemma}

\newtheorem{proposition}{Proposition}

\newtheorem{remark}{Remark}

\newcommand{\bs}[1]{{\bm{#1}}}

\newcommand{\trace}{{\textrm{Tr}}}

\begin{document}


\title{Generalized Low-Rank Optimization for Topological Cooperation  in Ultra-Dense Networks}
\author{Kai~Yang,~\IEEEmembership{Student Member,~IEEE,}
        Yuanming~Shi,~\IEEEmembership{Member,~IEEE,}
        and~Zhi~Ding,~\IEEEmembership{Fellow,~IEEE}

\thanks{K. Yang is with the School of Information Science and Technology, ShanghaiTech University, Shanghai, China, and also with the University of Chinese Academy of Sciences, Beijing, China (e-mail: yangkai@shanghaitech.edu.cn).}
\thanks{Y. Shi is with the School of Information Science and Technology, ShanghaiTech University, Shanghai, China (e-mail: shiym@shanghaitech.edu.cn).} \thanks{Z. Ding is with the Department of Electrical and Computer Engineering,
University of California at Davis, Davis, CA 95616 USA (e-mail:
zding@ucdavis.edu).}
}

\maketitle



\begin{abstract} 
Network densification is a natural way to support dense 
mobile applications under stringent requirements, such as ultra-low latency, ultra-high data rate, and massive connecting devices. 
Severe interference in ultra-dense networks poses a key bottleneck. 
Sharing channel state information (CSI) and messages across transmitters 
can potentially alleviate interferences and improve system performance. 
Most existing works on interference coordination require significant CSI signaling overhead and are impractical in ultra-dense networks. 
This paper investigate topological cooperation to manage interferences 
in message sharing based only on network connectivity information. 
In particular, we propose a generalized low-rank optimization approach to maximize achievable degrees-of-freedom (DoFs). To tackle the 
challenges of poor structure and non-convex rank function, we develop Riemannian optimization algorithms to solve 
a sequence of \textit{complex} fixed rank subproblems through a rank growth strategy. By exploiting the non-compact Stiefel manifold formed by the set of complex full column rank matrices, we develop Riemannian optimization algorithms to
solve the complex fixed-rank optimization problem by applying the semidefinite lifting technique and Burer-Monteiro factorization approach. Numerical results demonstrate the computational efficiency and  higher DoFs achieved by the proposed algorithms.
\end{abstract}


\section{Introduction}
The upsurge of wireless applications, including Internet-of-Things (IoT), Tactile Internet, tele-medicine and mobile edge artificial intelligence, 
is driving the paradigm shift of wireless networks from 
content delivery to skillset-delivery 
networks \cite{Fettweis_JSAC2016}.  Network densification \cite{shi2015largecooperative} has emerged as a promising approach to support innovative mobile applications with stringent requirements
such as ultra-low latency, ultra-high data rate and massive devices connectivity. Unfortunately, interference in
dense wireless network deployment becomes a key capacity limiting factor 
given large numbers of transmitters and receivers. 
Network cooperation through sharing channel state information (CSI) and messages among transmitting nodes is a viable technology to improve the spectral efficiency and energy efficiency in ultra-dense wireless networks.

Under shared CSI among transmitters, interference alignment \cite{cadambe2008interference} is shown to mitigate interferences base
on linear coding schemes, capable of achieving half the cake for each user in $K$-user interference channel. Cooperative transmission \cite{gesbert2010multi} with message sharing has shown to be able to further improve system throughput. In particular, through centralized signal processing and interference management with full message sharing via the cloud data center, cloud radio access network (Cloud-RAN) \cite{Yuanming_TWC2014} can harness the advantages of network densification. By pushing the storage resources to the network edge \cite{yang2016low_globalsip}, cache-aided wireless network \cite{maddah2014fundamental} provides a cost effective way to enable transmission cooperation.

Unfortunately, most existing works on network cooperation 
lead to significant channel signaling overhead. This is practically
challenging in ultra-dense networks. 
A growing body of recent works has
hence been focusing on CSI acquisition overhead reduction for interference coordination in wireless networks. Among them, delay effect
in CSI acquisition has been considered in \cite{maddah2012completely}. 
Both \cite{shi2015optimal} and \cite{razaviyayn2016stochastic} have studied
transceiver design using partial CSI, requiring instantaneous 
CSI for strong links and only
distribution CSI of the remaining weak links. 
In addition, finite precision CSI
feedback \cite{davoodi2017generalized} and the compressed channel estimation \cite{bajwa2010compressed} can further reduce 
CSI acquisition overhead. 

However, the applicability of the aforementioned results in practical 
systems remains unclear, which motivates a recent proposal on topological interference management (TIM) \cite{Jafar_TIT2013TIM}. The main idea of 
TIM is to manage the interference  based only on the network connectivity information, which can significantly reduce the CSI acquisition overhead. 
By requiring only network topologies, TIM becomes one of the most promising and powerful schemes for interference management in ultra-dense wireless networks. By further enabling message sharing, the work of \cite{yi2015topological} shows that transmitter cooperation based only on network topology information can strictly improve the degrees-of-freedom (DoFs). However, their results are only applicable to some specific network connectivity patterns.


In this paper, we propose a generalized low-rank optimization approach for investigating the benefits of topological cooperation for any network topology. We begin by first establishing the generalized interference alignment conditions based only on the network connectivity information 
with message sharing among transmitters. A low-rank model is further developed to maximize the achievable DoFs by exploiting the relationship between the model matrix rank and the achievable DoFs. The developed low-rank matrix optimization model thus generalizes the low-rank matrix completion model \cite{shi2016low} without message sharing among transmitters. Unfortunately, 
the resulting generalized low-rank optimization problem in complex field is non-convex and highly intractable due to poor structure, 
for which novel and efficient algorithms need to be developed.

Low-rank matrix optimization models have wide range of applications in machine learning, high-dimensional statistics, signal processing and wireless networks \cite{shi2016low,davenport2016overview,sridharan2015linear,papailiopoulos2012interference}. A wealth of recent works focus on both convex 
approximation and non-convex algorithms to solve the non-convex and highly intractable low-rank optimization problems. 
Nuclear norm is a well-known convex proxy for non-convex rank function with optimality guarantees under statistical models \cite{candes2009exact}. 
To further reduce the storage and computation overhead for low-rank optimization, non-convex approach based on matrix factorization shows
good promises \cite{jain2013low}. With suitable statistical models, the non-convex methods can also find globally optimal solution for some structured optimization problems such as matrix completion \cite{ge2016matrix}. In particular, the work of \cite{yi2016topological} adopted an alternating minimization algorithm to exploit topological transmitter cooperation gains. This algorithm stores the iterative
results in the factored form and optimizes over one factor 
while fixing the other.

Nevertheless, the nuclear norm based convex relaxation approach 
in fact fails to solve the formulation of generalized low-rank matrix optimization problem because of the poor structures. Actually, the nuclear norm minimization approach always yields a full-rank matrix solution. Alternating minimization \cite{jain2013low} algorithm by factorizing the fixed-rank matrix is particularly
useful when the resulting problem is biconvex with respect 
to the two factors in matrix factorization. However, the convergence
 of the alternating minimization algorithm heavily depends
on the initial points with slow convergence rates. 
It may also yield poor performance in achievable DoFs, as it 
only guarantees convergence to the first-order stationary points \cite{shi2016low,yi2016topological}. In contrast, Riemannian optimization \cite{Absil_2009optimizationonManifolds} approach has shown to be effective 
in improving the achievable DoFs by solving the low-rank matrix 
optimization problems, as the Riemannian trust-region algorithm guarantees convergence to the second-order stationary points with high precision solutions \cite{shi2016low}. Furthermore, the Riemannian optimization algorithms are robust to initial points in ensuring convergence \cite{boumal2016global} with fast convergence rates. However, no available Riemannian optimization algorithms have been developed for the general non-square low-rank problems in the 
\textit{complex field}.  In this work,
we develop Riemannian optimization algorithms for solving the presented generalized low-rank optimization problem in the complex field.

\subsection{Contributions}
In this paper, we develop a generalized interference 
alignment condition to enable transmitter cooperation  based
only on the network topology information.
We present a generalized low-rank model to maximize the achievable DoFs. 
To address the special challenges in the resulting generalized low-rank optimization problem, we develop Riemannian optimization algorithms by exploiting the non-compact Stiefel manifold of 
fixed-rank matrices in complex field. Specifically, we propose to solve the generalized low-rank optimization problem by solving a sequence of fixed rank subproblems with rank increase. By applying
semidefinite lifting technique \cite{ge2017no}, the fixed rank 
subproblem is reformulated as a positive semidefinite matrix problem 
in complex field with rank constraint. By applying the Burer-Monteiro \cite{burer2003nonlinear} parameterization approach to factorize the 
positive semidefinite matrix, the resulting problem turns out to be a
Riemannian optimization problem on complex non-compact Stiefel manifold. Therefore, the generalized low-rank optimization problem can be successfully solved by developing Riemannian optimization algorithms on the complex-valued non-compact Stiefel manifold. 

We  summarize the main contributions of this work as follows:
\begin{enumerate}
        \item We establish a generalized interference alignment condition  to enable transmitter cooperation with message sharing  based 
        only
        on network connectivity information. We develop a generalized low-rank model to maximize the achievable DoFs.
        \item We develop first-order and second-order Riemannian optimization algorithms for solving the generalized low-rank optimization problem in \emph{complex field}. We exploit
         the complex compact Stiefel manifold of complex fixed-rank matrices using the semidefinite lifting and Burer-Monteiro factorization techniques.
        \item Numerical results demonstrate that the proposed second-order Riemannian trust-region algorithm is able to achieve
        the highest DoFs with high precision second-order stationary point solutions. Furthermore, its computing time is comparable to the first-order Riemannian conjugate gradient algorithm in medium network sizes. Overall, the Riemannian algorithms show much better performance than the alternating minimization algorithm.
\end{enumerate}

\subsection{Organization and Notation}
The remainder of this paper is organized as follows. In Section \ref{sec:system}, we first introduce the system model, before
 establishing the generalized topological interference alignment conditions.
We develop the generalized low-rank model
in Section \ref{sec:GLRM}, 
and derive a positive semidefinite reformulation with the Burer-Monteiro approach to address the low-rank optimization problem in complex field.
We derive Riemannian algorithms on complex non-compact Stiefel manifold 
in Section \ref{sec:manifold}. Section \ref{sec:simulation}
provides simulation results. Finally, we conclude this work
in Section \ref{sec:conclusion}, 

We use $[K]$ denote the set $\{1,2,\cdots,K\}$. $\mathbb{S}_+^{N}$ denotes the set of all $N\times N$ Hermitian positive semidefinite matrices. And $\langle \cdot,\cdot \rangle$ denotes inner product, i.e., $\langle \bs{A}, \bs{X}\rangle=\trace(\bs{A}^{\sf{H}}\bs{X})$.

\section{System Model and Problem Formulation}\label{sec:system}
In this section, we establish the generalized interference alignment condition for partially connected $K$-user interference channel with transmitter cooperation.

\setlength\arraycolsep{2pt}
\subsection{System Model}
Consider a partially-connected interference channel with $K$ 
single-antenna transmitters and $K$ single-antenna receivers.
Transmitters aim to deliver
a set of independent messages $W_1,W_2$, $\cdots,W_K$ 
to receivers $1,\,2,\,\cdots,\,K$, respectively. 
Transmitter $k$ has message $W_k$ and is always connected with receiver $k$.
The channel coefficient $h_{kl}\in\mathbb{C}$ 
between the $l$-th transmitter and the $k$-th user is nonzero only for $(k,l)\in\mathcal{E}$. 
Block fading channel model is 
considered in this paper, i.e., $h_{kl}$ remains stationary 
in $r$ consecutive 
channel uses, during which the input-output relationship is given by
\begin{eqnarray}
\bs{y}_k=\sum_{(k,i)\in\mathcal{E}}{h}_{ki}{\bs{x}}_i+\bs{z}_k, ~\forall k\in[K],
\end{eqnarray}
where ${\bs{x}}_{i}\in\mathbb{C}^{r}$ is the transmitted signal at transmitter $i$, $\bs{y}_k\in\mathbb{C}^{r}$ is the received signal 
at receiver $k$, and $\bs{z}_k\in\mathbb{C}^{r}$ is the additive 
isotropic white Gaussian noise, i.e., $\bs{z}_k\sim\mathcal{CN}(\bs{0},\bs{\Sigma}_k)$ with 
$\bs{\Sigma}_k\in\mathbb{C}^{r\times r}$. Partial connectivity 
of the interference channel provides opportunities to enable 
cooperative transmission based 
only on the network connectivity information. Specifically, transmitter 
cooperation is enabled with message sharing, for which we denote the index set of 
messages available at transmitter $k$ as $\mathcal{S}_k\subseteq[K]$. A 
$5$-user example of such system is shown in Fig. \ref{fig:system}.

Let $R(W_k)$ be the achievable data rate of message $W_k$, i.e., 
there exists a coding scheme such that the rate of message $W_k$ 
is $R(W_k)$ and the decoding error probability can be arbitrarily small.
Let $\textrm{SNR}$ denote the signal-to-noise-ratio. 
For each message delivery, the degree-of-freedom (DoF) \cite{Jafar_TIT2013TIM}, the first order characterization of channel capacity, is defined as
\begin{equation}
        \textrm{DoF}(W_k) = \lim_{\textrm{SNR}\rightarrow\infty}\frac{R(W_k)}{\log(\textrm{SNR})},~\forall k\in[K].
\end{equation}
The set of achievable DoF allocation is denoted as $\{\textrm{DoF}(W_1),\cdots,\textrm{DoF}(W_K)\}$, 
whose closure is called the DoF region. 
This paper adopts DoF as the performance metric
and designs a linear coding scheme.  

\subsection{Linear Coding Strategy}
Linear coding scheme is attractive for interference management owing
to its low complexity. Specifically, its optimality in terms 
of DoF has been shown via interference alignment \cite{cadambe2008interference}. Its effectiveness has been demonstrated in the problems of topological interference management (TIM) and index coding \cite{Jafar_TIT2013TIM}. We thus focus on linear coding scheme 
to design low complexity and efficient approaches for maximizing
achievable DoFs. 

Suppose each message $W_k$ is represented by a complex vector $\bs{s}_k\in\mathbb{C}^{d_k}$ with $d_k$ data streams. Let ${\bs{V}}_{kj}\in\mathbb{C}^{r\times d_j}$ be the precoding matrix at transmitter $k$ for message $W_j$. Then the
transmitted signal is given by
\begin{equation}
        \bs{x}_k = \sum_{j\in\mathcal{S}_k}\bs{V}_{kj}\bs{s}_j.
\end{equation}
Consequently, the received signal at user $k$ is 
\begin{equation}
\bs{y}_k = \!\!\!\!\!\!\!\!\!\!\!\!\sum_{j:(k,j)\in\mathcal{E},k\in\mathcal{S}_j}\!\!\!\!\!\!\!\!\!\! h_{kj}\bs{V}_{jk}\bs{s}_k+ \sum_{i\ne k}\sum_{j:(k,j)\in\mathcal{E},i\in\mathcal{S}_j}\!\!\!\! \!\!\! h_{kj}\bs{V}_{ji}\bs{s}_i+\bs{z}_k.
\end{equation}
We let $\bs{U}_k\in\mathbb{C}^{r\times d_k}$ be the decoding matrix at receiver $k$. 
\begin{figure}
        \centering
        \includegraphics[width=0.65\columnwidth]{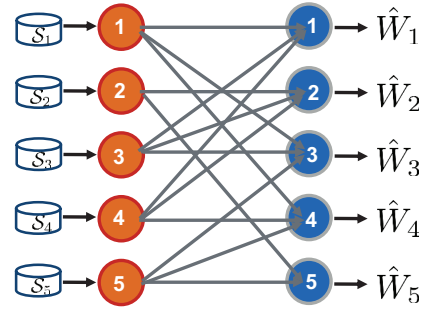}
        \caption{The architecture of the partially-connected $K$-user interference channel with transmitter cooperation. $\mathcal{S}_i$ denotes the index set of messages available at transmitter $i$.}\label{fig:system}
\end{figure}

In densified wireless networks, interference is a key bottleneck to support high data rate and low latency. To alleviate interferences by aligning the intersection of interference spaces, the following interference alignment conditions were presented in \cite{Jafar_TIT2013TIM,bresler2014feasibility}
\begin{eqnarray}
\label{equa:IA1}
\sum\limits_{j:(k,j)\in\mathcal{E},k\in\mathcal{S}_j}h_{kj}\bs{U}_k^{\sf{H}}\bs{V}_{jk} &\ne &0,~\forall k\in [K], \\
\label{equa:IA2}
\sum\limits_{j:(k,j)\in\mathcal{E},i\in\mathcal{S}_j}h_{kj}\bs{U}_k^{\sf{H}}\bs{V}_{ji} &=&0,~i\ne k.
\end{eqnarray}
Correspondingly, the message at receiver $k$ is decoded via
\begin{equation}
\hat{\bs{s}}_k=(\!\!\!\!\!\!\sum\limits_{j:(k,j)\in\mathcal{E},k\in\mathcal{S}_j}\!\!\!\!\!\!h_{kj}\bs{U}_k^{\sf{H}}\bs{V}_{jk})^{-1}\bs{U}_k^{\sf{H}}\bs{y}_k.
\end{equation}
If there exists $\bs{U}_k$'s, $\bs{V}_{ji}$'s satisfying interference alignment conditions (\ref{equa:IA1}) and (\ref{equa:IA2}), DoF tuple $({d_1}{r}^{-1},\cdots,{d_K}{r}^{-1})$ is then achievable. We thus can achieve the highest DoF by finding the minimal channel use number $r$. 

\subsection{Topology-Based Alignment Condition}
Note that equations (\ref{equa:IA1}) and (\ref{equa:IA2}) are always feasible by increasing $r$. 
However, the interference alignment conditions (\ref{equa:IA1}) and (\ref{equa:IA2}) require the knowledge of channel coefficients $h_{ij}$s at the transmitters. In practice, obtaining dense network channel state information (CSI)
at transmitters often requires large signaling overhead, 
which presents a severe obstacle to their application in
densified wireless networks. 
One desirable way to address the CSI acquisition overhead issue is to 
establish new interference alignment conditions based only on the network connectivity information, for which we present the following generalized interference alignment conditions for \textit{topological cooperation}
\begin{eqnarray}
\label{siso_cgm1}
&&\textrm{det}\left(\sum_{j:(k,j)\in\mathcal{E},k\in\mathcal{S}_j}{\bs{U}}_k^{\sf{H}}{\bs{V}}_{jk}\right)\ne
0,~~\forall k\in [K],\\
\label{siso_cgm2}
&&\bs{U}_k^{\sf{H}}\bs{V}_{ji} =\bs{0},~~i\in\mathcal{S}_j,i\ne k, (k,j)\in\mathcal{E}.
\end{eqnarray}
Here, ``topological cooperation" refers to the fact that for cooperation enabled transmitters, we design transceivers to manage interferences based on network topology information instead of instantaneous channel state information. 
\begin{proposition}
        For generic channel coefficients $h_{ij}$'s randomly distributed according
 to some continuous probability distribution \cite{razaviyayn2012degrees}, if (\ref{siso_cgm1}) and (\ref{siso_cgm2}) hold for some $\bs{U}_{k},\bs{V}_{ji}$s based only on the network topology information, then they shall satisfy the channel dependent interference alignment conditions (\ref{equa:IA1}) and (\ref{equa:IA2}) with probability 1.
\end{proposition}
\begin{proof}
Let $\bs{Z}_{1},\cdots,\bs{Z}_T\in\mathbb{C}^{d\times d}$ denote the set of matrices $\{\bs{U}_k^{\sf{H}}\bs{V}_{jk}:(k,j)\in\mathcal{E},k\in\mathcal{S}_j\}$ given $k$. Our goal is to prove that if $\textrm{det}\Big(\sum_{t=1}^{T}\bs{Z}_{t}\Big)\ne 0$, the probability of $\textrm{det}\Big(\sum_{t=1}^{T}h_t\bs{Z}_{t}\Big)\ne 0$ is $1$ for generic $h_1,\cdots,h_T$. Note that $\sum_{t=1}^{T}\bs{Z}_{t} = \begin{bmatrix}
        \bs{Z}_1 & \cdots & \bs{Z}_T
    \end{bmatrix}\begin{bmatrix}
        \bs{I} & \cdots & \bs{I}
    \end{bmatrix}^H$. Thus, the condition $\textrm{det}\Big(\sum_{t=1}^{T}\bs{Z}_{t}\Big)\ne 0$ implies that $\begin{bmatrix}
        \bs{Z}_1 & \cdots & \bs{Z}_T
    \end{bmatrix}$ has full rank, i.e., the dimension of $\textrm{span}\{\bs{Z}_j\}$ is $r$. Since the solution to the determinant equation $\textrm{det}\Big(\sum_{t=1}^{T}h_t\bs{Z}_{t}\Big)= 0$ is an algebraic hypersurface \cite{bresler2014feasibility}, the probability of the linear combination $\sum_{t=1}^{T}h_t\bs{Z}_{t}$ with generic coefficients $h_1,\cdots,h_T$ lying on the algebraic hypersurface is hence zero. Therefore, $\textrm{det}\Big(\sum_{t=1}^{T}h_t\bs{Z}_{t}\Big)\ne 0$ holds with probability $1$.
\end{proof}

By leveraging conditions (\ref{siso_cgm1}) and (\ref{siso_cgm2}), interferences can be aligned based only on network topology CSI instead of full 
CSI. 
This significantly reduces the overhead of CSI acquisition. In particular, the topological interference alignment condition without message sharing is given by \cite{shi2016low}
\begin{eqnarray}
\label{TIM_con1}
&&\textrm{det}\left({\bs{U}}_k^{\sf{H}}{\bs{V}}_{k}\right)\ne
0,~~\forall k\in [K],\\
\label{TIM_con2}
&&\bs{U}_k^{\sf{H}}\bs{V}_{j} =\bs{0},~~j\ne k, (k,j)\in\mathcal{E},
\end{eqnarray}
which is a special case of (\ref{siso_cgm1}) and (\ref{siso_cgm2}) with $\mathcal{S}_j=\{j\}$ and ${\bs{V}}_{j}$ denoting as ${\bs{V}}_{jj}$. Conditions (\ref{siso_cgm1}) and (\ref{siso_cgm2}) thus manifest the benefits of transmitter cooperation, as solutions to (\ref{TIM_con1}) and (\ref{TIM_con2}) are always solutions to (\ref{siso_cgm1}) and (\ref{siso_cgm2}), but not conversely.

\begin{remark}
    This work assumes that there are equal number of transmitters and receivers. Nevertheless, the principle applies for arbitrary number of transmitters and receivers. This is because both Proposition 1 and the low-rank matrix representation for precoding and decoding matrices in Section 3 hold for any number of transmitters and receivers. For simplicity of notation we consider a system with $K$ transmitters and receivers in this paper.
\end{remark}

\section{Generalized Low-Rank Optimization for Topological Cooperation}\label{sec:GLRM}
This section develops a generalized low-rank optimization framework to maximize achievable DoFs under topological cooperation. To address the challenges of the present generalized low-rank optimization problem in complex field and 
to exploit the algorithmic benefits of Riemannian optimization, we propose to reformulate an optimization problem over the complex non-compact Stiefel manifold by using the semidefinite lifting and Burer-Monteiro approaches.

\subsection{Generalized Low-Rank Model for Topological Cooperation}
Without loss of generality, we restrict $\sum_{j:(k,j)\in\mathcal{E},k\in\mathcal{S}_j}{\bs{U}}_k^{\sf{H}}{\bs{V}}_{jk}=\bs{I}$ in condition (\ref{siso_cgm1}). By letting $m=\sum_{k}d_k,~ n=K\sum_{k}d_k$ and defining
\begin{align*}
         \bs{U} &= \begin{bmatrix}
                        \bs{U}_1 & \cdots & \bs{U}_{K}
                \end{bmatrix}\in\mathbb{C}^{r\times m}, \\
                \bs{V}_{j}&=[{\bs{V}}_{j1},\cdots,{\bs{V}}_{jK}]\in\mathbb{C}^{r\times m}, \\
                \bs{V}&= \begin{bmatrix}
                        \bs{V}_1 & \cdots & \bs{V}_{K}
                \end{bmatrix}\in\mathbb{C}^{r\times n},\\
                \bs{X}&=[\bs{X}_{kj}^{i}]=[\bs{U}_k^{\sf{H}}\bs{V}_{ji}]={\bs{U}}^{\sf{H}}{\bs{V}}\in\mathbb{C}^{m\times n},
\end{align*}
the rank of matrix $\bs{X}$ is given as
\begin{equation}
        \textrm{rank}(\bm{X})=r = d_k/\text{DoF}(W_k).
\end{equation}
We thus can maximize the achievable DoF for interference-free message delivery by solving the following \emph{generalized low-rank optimization} problem
\begin{eqnarray}
\mathscr{P}:\mathop{\textrm{minimize}}_{{\bs{X}}\in\mathbb{C}^{m\times n}}&& \textrm{rank}({\bs{X}}) \nonumber\\
\textrm{subject to}&& \mathcal{A}({\bs{X}})={\bs{b}},\label{GLRM_rank}
\end{eqnarray}
where the affine constraint $\mathcal{A}({\bs{X}})={\bs{b}}$ captures
\begin{eqnarray}
&& \sum_{j:(k,j)\in\mathcal{E},k\in\mathcal{S}_j}\bs{X}_{kj}^{k}=\bs{I},~\forall k\in [K] \label{prob:GLRM1eq1}\\
&&\bs{X}_{kj}^{i} =\bs{0},~i\ne k,i\in\mathcal{S}_j, (k,j)\in\mathcal{E} \label{prob:GLRM1eq2}
\end{eqnarray}
and $\mathcal{A}:\mathbb{C}^{m\times n}\mapsto\mathbb{C}^{l}$.

For the simpler case without message sharing, the topological interference alignment problem can be formulated as the following low-rank matrix completion problem \cite{shi2016low,hassibi2014topological}
\begin{eqnarray}\label{prob:TIM}
\mathop{\textrm{minimize}}_{\bm{X}\in\mathbb{C}^{m\times m}}&& \textrm{rank}(\bs{X}) \nonumber\\
\textrm{subject to}&& \bs{X}_{kk}=\bs{I},~\forall k\in [K] \nonumber\\
&&\bs{X}_{kj} = \bs{0},~j\ne k, (k,j)\in\mathcal{E},
\end{eqnarray}
which is a special case of problem $\mathscr{P}$. The resulting low-rank matrix completion model is demonstrated in Fig. \ref{fig:tim_lrmc}. 
\begin{figure}[t]
        \centering
        \subfloat[TIM]{\includegraphics[width=0.47\columnwidth]{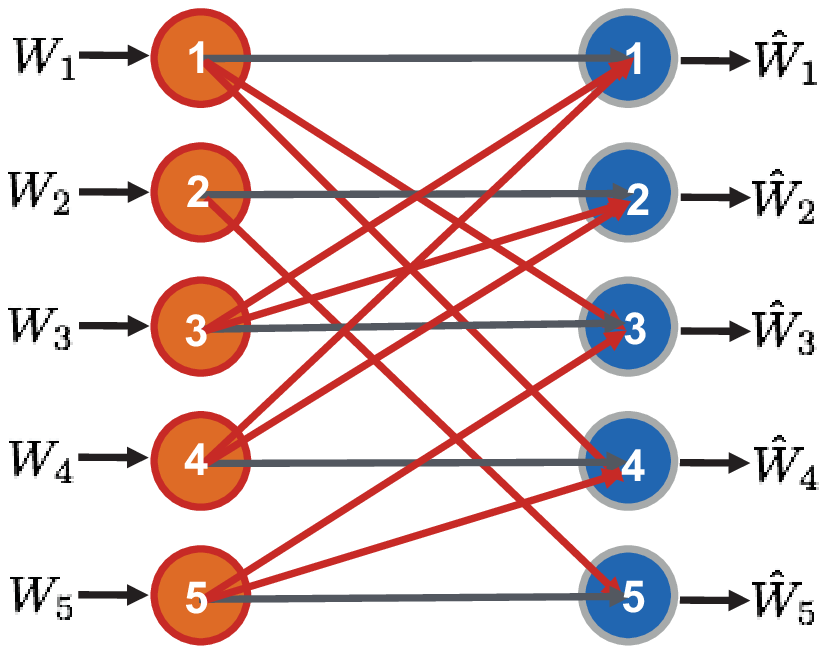}\label{fig:TIM}} \hfil
        \subfloat[Matrix Completion Model]{\includegraphics[width=0.45\columnwidth]{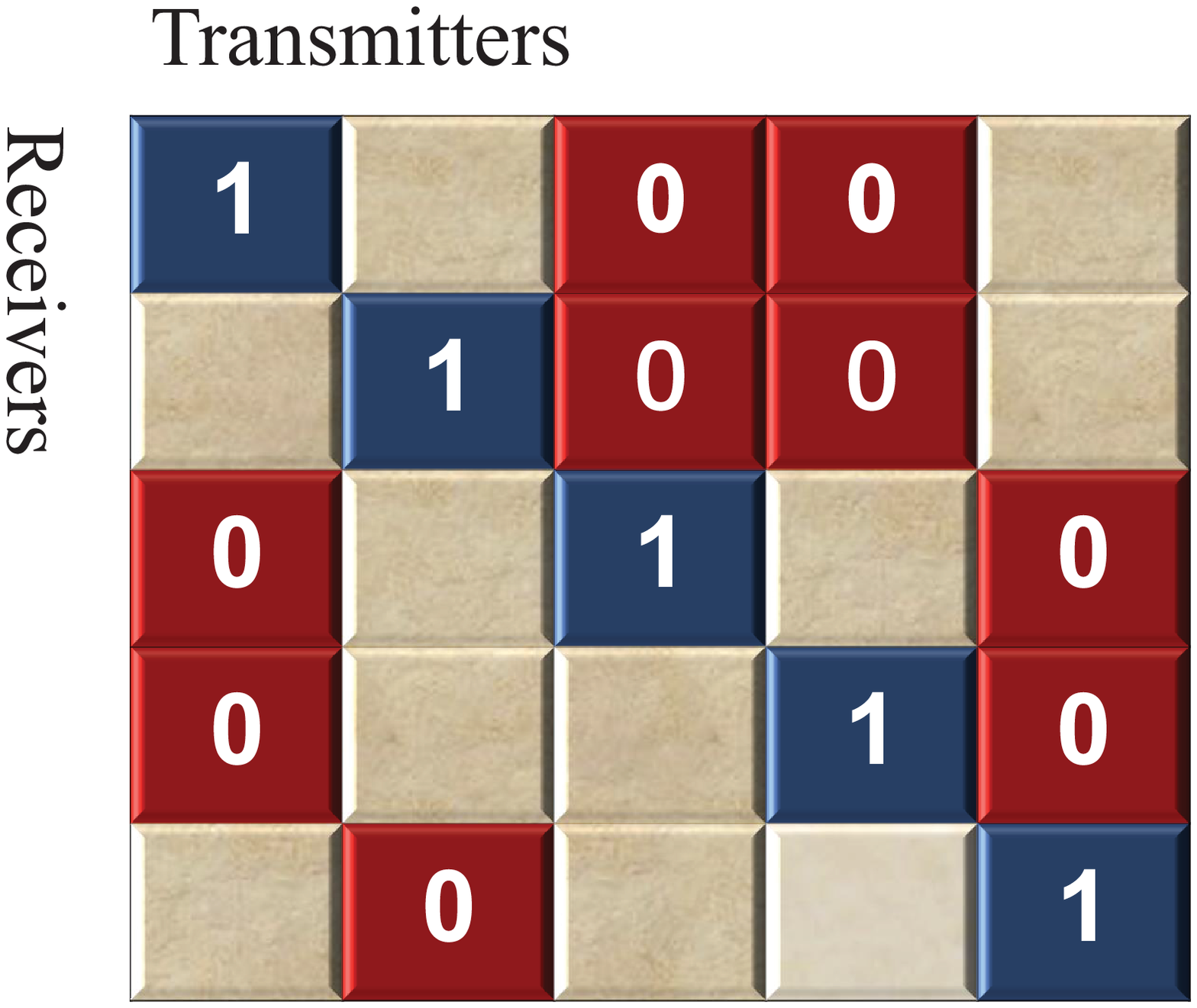}\label{fig:LRMC}}
        \caption{Matrix completion model for the topological interference alignment  without message sharing for single data stream $d_k=1$. In this case, all the diagonal entries are set to 1 as $\mathcal{S}_k=\{k\}$. As an example, the $(2,3)$-th entry is zero because the second receiver is connected with the third transmitter as interference. And the $(3,2)$-th entry of the matrix can be arbitrary value as $(3,2)\not\in\mathcal{E}$.}
        \label{fig:tim_lrmc}
\end{figure}
\subsection{Problem Analysis}
Basically, methods for solving low-rank problems can be divided into two categories. One uses convex relaxation approach and the other one uses nonconvex approach based on matrix factorization. In addition, penalty decomposition method is proposed in \cite{zhang2011penalty} for low-rank optimization problems. The inner iterations adopt a block coordinated descent method, whereas the outer iterations update the weight of rank function. However, each inner iteration requires the computation of singular value decomposition, which leads to large computation overhead ($\mathcal{O}(mnl+m^2n+m^3)$). Therefore, it is not suitable for our tranceiver design problem in ultra-dense networks.

\subsubsection{Convex Relaxation Methods}
Nuclear norm is a well-known convex proxy \cite{davenport2016overview} for rank function. 
The nuclear norm relaxation approach for problem $\mathscr{P}$ is given by
\begin{eqnarray}
\mathop{\textrm{minimize}}_{{\bs{X}}\in\mathbb{C}^{m\times n}}&& \|{\bs{X}}\|_* \nonumber\\
\textrm{subject to}&& \mathcal{A}({\bs{X}})={\bs{b}}.\label{prob:nuc}
\end{eqnarray}
It can be solved by an equivalent semidefinite programming (SDP) problem
\begin{eqnarray}
\mathop{\textrm{minimize}}_{\bs{X},\bs{W}_1,\bs{W}_2}&& \textrm{Tr}(\bs{W}_1)+\textrm{Tr}(\bs{W}_2) \nonumber\\
\textrm{subject to}&& \mathcal{A}({\bs{X}})={\bs{b}}, \label{prob:nuc} \\
&& \begin{bmatrix}
  \bs{W}_1 & \bs{X} \\ \bs{X}^{\sf{H}} & \bs{W}_2
\end{bmatrix}\succeq \bs{0} \nonumber.
\end{eqnarray} 
Unfortunately, the SDP solution requires computing singular value decomposition at each iteration, which is not scalable to large problem sizes in ultra-dense networks. Specifically, with high precision second-order interior point method, the convergence rate is fast while the computational cost for each iteration is $\mathcal{O}((mn+l)^3)$ due to computing the Newton step \cite{boyd2004convex}. Using the first-order algorithm alternating direction method of multipliers (ADMM) \cite{o2016conic}, the computational cost is $\mathcal{O}(mnl+m^2n+m^3)$ at each iteration. Furthermore, the nuclear norm relaxation approach always yields a full rank solution due to the poor structure of the affine operator $\mathcal{A}$. 
\begin{proposition}\label{proposition:nuc}
      The nuclear norm relaxation approach   (\ref{prob:nuc}) for the generalized low-rank optimization problem $\mathscr{P}$  always yields a full rank solution.
\end{proposition}
\begin{proof}
        See Appendix \ref{append:nuc}.
\end{proof}
Therefore, the nuclear norm relaxation based approach is inapplicable for the poorly structured low-rank optimization problem $\mathscr{P}$. We thus call problem $\mathscr{P}$ as the generalized low-rank optimization problem.
\subsubsection{Nonconvex Approaches}
A rank $r$ matrix ${\bs{X}}$ can be factorized as ${\bs{X}}={\bs{L}}{\bs{R}}^{\sf{H}}$, where ${\bs{L}}\in\mathbb{C}^{m\times r}$ and ${\bs{R}}\in\mathbb{C}^{n\times r}$. Nonconvex approaches to low-rank optimization leverage matrix factorizations and design various updating strategies for two factors $\bm{U}$ and $\bm{V}$. By solving a sequence of the fixed rank least square subproblems based on matrix factorization
\begin{eqnarray}\label{GLRM_fixedrank}
\mathop{\textrm{minimize}}_{{\bs{X}}\in\mathbb{C}^{m\times n}} && f_0({\bs{X}})=\frac{1}{2}\|\mathcal{A}({\bs{X}})-{\bs{b}}\|_2^2 \nonumber\\
\textrm{subject to}&& \textrm{rank}({\bs{X}})=r,
\end{eqnarray}
and increasing $r$, we can find the minimal rank $r$ for the original problem $\mathscr{P}$.

Specifically, for the rank constrained problem (\ref{GLRM_fixedrank}) with convex objective function, the alternating minimization \cite{jain2013low} algorithm can function as follows:
\begin{eqnarray}
\bm{L}_{k+1} &=& \arg\min_{{\bs{L}}}~~ f_0(\bs{L}\bs{R}_{k}^{\sf{H}}) \label{GLRM_fixedrank_altmin1},\\
\bm{R}_{k+1} &=& \arg\min_{{\bs{R}}}~~ f_0(\bs{L}_{k+1}\bs{R}^{\sf{H}}). \label{GLRM_fixedrank_altmin2}
\end{eqnarray}
It essentially optimizes the bi-convex objective function $f_0(\bm{L}\bm{R}^{\sf{H}})$ by freezing one of ${\bs{L}}$ and ${\bs{R}}$ alternatively. However, the convergence of the alternating minimization algorithm are sensitive to initial points and its convergence rate 
can be slow. Furthermore, it may yield poor performance for achievable DoFs maximization by only converging to first-order stationary point.

In contrast, Riemannian optimization algorithms are capable 
of updating the two factors $\bm{L}$ and $\bm{R}$ 
\emph{simultaneously} by exploiting the 
quotient manifold geometry of fixed-rank matrices based 
on matrix factorization. First-order Riemannian conjugate 
gradient and second-order Riemannian trust-region algorithm 
can help find first-order stationary points and second-order 
stationary points, respectively. It has been shown in \cite{boumal2016global} that Riemannian optimization algorithms converge to first-order and second-order stationary points from arbitrary initial points. 
Furthermore, Riemannian trust-region algorithm can 
achieve high achievable DoFs with second-order stationary points 
while also enjoys locally super-linear \cite{Absil_2009optimizationonManifolds} convergence rates.

\begin{remark}
    The invariance of matrix factorization $(\bs{L}\bs{M}^{\sf{H}},\bs{M}^{-1}\bs{R})$ for any full rank matrix $\bs{M}$ makes the critical points of $f_0$ parameterized with $\bs{L}$ and $\bs{R}$ are not isolated in Euclidean space. This indeterminancy profoundly affects the convergence of second-order optimization algorithms \cite{shi2017useradmission,journee2010lowrank}. To address this issue, we shall develop efficient algorithms on the quotient manifold instead of Euclidean space.
\end{remark}

Unfortunately, available Riemannian algorithms for non-square fixed-rank matrix optimization problems only operate in $\mathbb{R}$ field \cite{Absil_2009optimizationonManifolds} and do not directly apply to
solve problem (\ref{GLRM_fixedrank}) in the complex filed. 
Inspired by the fact that the complex non-compact Stiefel manifold is well defined \cite{yatawatta2013radio}, we propose to reformulate complex matrix optimization problem (\ref{GLRM_fixedrank}) on the complex non-compact Stiefel manifold by using the Burer-Monteiro approach. 
Specifically, applying semidefinite lifting, the original problem (\ref{GLRM_fixedrank}) is equivalently reformulated into rank constrained positive semidefinite matrix optimization, before factorizing the semidefinite matrices using the Burer-Monteiro approach. The original complex matrix optimization problem (\ref{GLRM_fixedrank}) is thus reformulated as the Riemannian optimization problem over the well-defined non-compact Stiefel manifold in complex field.

\subsection{Semidefinite Lifting and the Burer-Monteiro Approach}\label{sec:liftBM}
Burer-Monteiro approach is a well-known nonconvex parameterization method for solving positive semidefinite (PSD) matrices problems \cite{burer2003nonlinear}. A rank $r$ PSD matrix ${\bs{Z}}\in\mathbb{S}^{N}$ can be factorized as ${\bs{Z}}={\bs{Y}}{\bs{Y}}^{\sf{H}}$ with ${\bs{Y}}\in\mathbb{C}^{N\times r}$. The linear operator $\mathcal{A}$ can be represented as a set of matrices $\bs{A}_i\in\mathbb{C}^{n\times m}$, i.e., 
\begin{equation}\label{eq:defA}
        \mathcal{A}({\bs{X}}) = [\langle \bs{A}_i, \bs{X}\rangle], i=1,\cdots,l.
\end{equation}
Then the objective function in  (\ref{GLRM_fixedrank}) can be rewritten as
\begin{equation}
        f_0(\bs{X})=\frac{1}{2}\sum_{i=1}^{l}|\langle \bs{A}_i, \bs{X}\rangle-b_i|^2.
\end{equation}
By semidefinite lifting \cite{ge2017no} ${\bs{X}}$ to
\begin{equation}
        {\bs{Z}}=\begin{bmatrix}
                {\bs{Z}}_{11} & {\bs{Z}}_{12} \\
                {\bs{Z}}_{21} & {\bs{Z}}_{22}
        \end{bmatrix}:=\begin{bmatrix}
                {\bs{L}}{\bs{L}}^{\sf{H}} & {\bs{L}}{\bs{R}}^{\sf{H}} \\
                {\bs{R}}{\bs{L}}^{\sf{H}} & {\bs{R}}{\bs{R}}^{\sf{H}}
        \end{bmatrix},
\end{equation}
we can reformulate problem (\ref{GLRM_fixedrank}) as a complex PSD matrix problem with rank constraint:
\begin{eqnarray}\label{GLRM_fixedrank_PSD}
\mathop{\textrm{minimize}}_{{\bs{Z}}\in\mathbb{S}_+^{N}}&& \frac{1}{2}\|\mathcal{B}({\bs{Z}})-{\bs{b}}\|_2^2  \nonumber\\
\textrm{subject to}&& \textrm{rank}({\bs{Z}})= r,
\end{eqnarray}
where $N=m+n$ and
\begin{equation}
        \mathcal{B}(\bs{Z})=\mathcal{A}(\bs{Z}_{12})=\mathcal{A}({\bs{L}}{\bs{R}}^{\sf{H}}).
\end{equation}
Here we use $\mathcal{B}$ to denote 
a set of matrices $\bs{B}_i\in\mathbb{C}^{N\times N}$
\begin{equation}\label{eq:defB}
        \bs{B}_i = \begin{bmatrix}
                \bs{0} & \bs{A}_i \\ \bs{0} & \bs{0}
        \end{bmatrix}, \langle \bs{B}_i, \bs{Z}\rangle = \langle \bs{A}_i, \bs{X}\rangle.
\end{equation}

We define ${\bs{Y}} = \begin{bmatrix}
{\bs{L}} \\ {\bs{R}}
\end{bmatrix}\in\mathbb{C}^{N\times r}$. 
The search space $\{\bm{Z}:\bm{Z}\in\mathbb{S}_+^{N},\textrm{rank}({\bs{Z}})= r\}$ admits a well-defined manifold structure, by factorizing $\bm{Z}=\bm{Y}\bm{Y}^{\sf{H}}$ based on the principles of Burer-Monteiro approach. Problem (\ref{GLRM_fixedrank_PSD}) thus can be transformed as
\begin{eqnarray}\label{GLRM_BM}
\mathop{\textrm{minimize}}_{{\bs{Y}}\in\mathbb{C}_*^{N\times r}}&& f({\bs{Y}})=\frac{1}{2}\|\mathcal{B}(\bs{Y}\bs{Y}^{\sf{H}})-\bs{b}\|_2^2.
\end{eqnarray}
This is a Riemannian optimization problem with a smooth ($C^{\infty}$) objective function over the complex \textit{non-compact Stiefel manifold} $\mathbb{C}_*^{N\times r}$, i.e., the set of all $N\times r$ full column rank matrices in complex field.

In summary, we propose to solve the generalized low-rank optimization 
problem by solving a sequence of \textit{complex} fixed-rank optimization problem using the Riemannian optimization technique. This is 
achieved by lifting the complex fixed-rank optimization problem 
into the complex positive semidefinite matrix optimization problem, 
followed by parameterizing it using the Burer-Monteior approach. 
This yields the Riemannian optimization problem over complex 
non-compact Stiefel manifold. After obtaining a solution $\bs{Y}$ from (\ref{GLRM_BM}), we can recover the solution 
$\bs{X}={\bs{L}}{\bs{R}}^{\sf{H}}$ to the original problem (\ref{GLRM_fixedrank}). The whole algorithm of addressing the transmitter cooperation problem based only on the network topology information is demonstrated in Algorithm \ref{algorithm:riemalgo}.

\SetNlSty{textbf}{}{:}
\IncMargin{1em}
\begin{algorithm}[tb]
\textbf{Input:} $\{S_j\},\mathcal{E},K,\{d_k\}$, accuracy $\varepsilon$. \\
Construct $\mathcal{B}$ and $\bm{b}$ following (\ref{prob:GLRM1eq1}) (\ref{prob:GLRM1eq2}) (\ref{eq:defA}) (\ref{eq:defB}). Let $N=m+n=(K+1)\sum_{k}d_k$.
 \\
 \For{$r=1,\cdots,N$}{
  Solve (\ref{GLRM_BM}) with Riemannian optimization algorithm.\\
  \If{$f({\bs{Y}}^{[r]})< \varepsilon$}{\Return $\bs{Y}^{[r]}$}
  }
 
 \textbf{Output:} $\bm{Y}^{[r]}$ and rank $r$.
 \caption{Optimization Framework for Transmitter Cooperation Based on Network Topology Information}
 \label{algorithm:riemalgo}
\end{algorithm}

\section{Matrix Optimization on Complex Non-compact Stiefel Manifold}\label{sec:manifold}
In this section, we shall develop Riemannian conjugate gradient and Riemannian trust-region algorithms for solving problem (\ref{GLRM_BM}). Riemannian optimization generalizes the concepts of gradient and Hessian in Euclidean space to Riemannian gradient and Hessian on manifolds. They are represented in the tangent space, which is the linearization of the search space.
\subsection{Quotient Geometry of Fixed-Rank Problem}
For problem (\ref{GLRM_BM}), the optima are not isolated because of ${\bs{Y}}{\bs{Y}}^{\sf{H}}$ remains invariant under the \textit{canonical projection} \cite[Sec 3.4.1]{Absil_2009optimizationonManifolds}
\begin{equation}\label{eq:canonicalprojection}
        \pi:~{\bs{Y}}\mapsto {\bs{Y}}{\bs{Q}}
\end{equation}
for any unitary matrix ${\bs{Q}}\in\mathcal{U}(r)$
where $\mathcal{U}(r)$ denotes the set of
$r\times r$ unitary matrices. 
To address this non-uniqueness we consider problem (\ref{GLRM_BM}) over the equivalent class
\begin{equation}
        [{\bs{Y}}] = \{{\bs{Y}}{\bs{Q}}:{\bs{Y}}\in\overline{\mathcal{M}}=\mathbb{C}_*^{N\times r},{\bs{Q}}\in\mathcal{U}(r)\},
\end{equation}
that is
\begin{eqnarray}\label{GLRM_equivalentclass}
\mathop{\textrm{minimize}}_{[{\bs{Y}}]\in\mathcal{M}}&& f([{\bs{Y}}]).
\end{eqnarray}
Then the whole set of feasible solutions can be represented by isolated points in the \textit{quotient manifold}, i.e. $\mathcal{M}=\overline{\mathcal{M}}/\sim
:=\overline{\mathcal{M}}/\mathcal{U}(r)$ with \textit{canonical projection} \cite[Sec 3.4]{Absil_2009optimizationonManifolds} $\pi$. Here $\sim$ is the equivalence relation and $\overline{\mathcal{M}}/\sim:=\{[\bs{Y}]:\bs{Y}\in\overline{\mathcal{M}}\}$. $\overline{\mathcal{M}}$ is considered as an abstract manifold.

\subsection{Riemannian Ingredients for Iterative Algorithms on Riemannian Manifolds}
By studying the unconstrained problem on the quotient manifold instead of the constrained problem in Euclidean space, Riemannian optimization can exploit the non-uniqueness of matrix factorization with Burer-Monteiro approach. We now develop conjugate gradient and trust-region algorithms on the Riemannian manifold. 
To achieve this goal, we first linearize the search space, 
by defining the concept of \textit{tangent space} \cite[Sec 3.5]{Absil_2009optimizationonManifolds} and associated ``inner product'' on the tangent space. Next, we will derive the expressions for Riemannian gradient and Riemannian Hessian in this subsection.

Specifically, tangent space $\mathcal{T}_{\bs{Y}}\overline{\mathcal{M}}$ is a vector space consisting of all tangent vectors to $\overline{\mathcal{M}}$ at $\bs{Y}$.
\begin{proposition}
        The tangent space of $\overline{\mathcal{M}}=\mathbb{C}_*^{N\times r}$ at $\bs{Y}$ is given by $\mathcal{T}_{\bs{Y}}\overline{\mathcal{M}}=\mathbb{C}^{N\times r}$.
\end{proposition}
\begin{proof}
        $\overline{\mathcal{M}}$ is an open submanifold \cite[Sec 3.5.2]{Absil_2009optimizationonManifolds} of $\mathbb{C}^{N\times r}$ and hence,  $\mathcal{T}_{\bs{Y}}\overline{\mathcal{M}}=\mathbb{C}^{N\times r}$ for all $\bs{Y}\in\overline{\mathcal{M}}$.
\end{proof}
In order to eliminate the non-uniqueness along the equivalent class $[\bs{Y}]$, we will decompose the 
tangent space into two orthogonal parts, i.e., \textit{vertical space} and \textit{horizontal space}. Vertical space $\mathcal{V}_{\bs{Y}}$ is the tangent space of equivalent class $[\bs{Y}]$, while horizontal space $\mathcal{H}_{\bs{Y}}$ is the orthogonal complement of vertical space in
the tangent space. That is,
\begin{equation}
        \mathcal{T}_{\bs{Y}}\overline{\mathcal{M}} = \mathcal{V}_{\bs{Y}} \oplus \mathcal{H}_{\bs{Y}},
\end{equation}
where $\oplus$ denotes the direct sum of two subspace. In this way, we can always find the unique ``lifted'' representation of the tangent vectors of $\mathcal{T}_{[\bs{Y}]}\mathcal{M}$ in $\mathcal{T}_{\bs{Y}}\overline{\mathcal{M}}$ at any element of $[\bs{Y}]$, i.e., for any $\bs{\xi}\in\mathcal{T}_{\bs{Y}}\mathcal{M}$ we define a unique horizontal lift $\overline{\bs{\xi}}\in\mathcal{H}_{\bs{Y}}$ at $\bs{Y}$ such that
\begin{equation}
         \overline{\bs{\xi}}:= \Pi_{\bs{Y}}^{h}\bs{\xi},
\end{equation}
where \textit{horizontal projection} $\Pi_{\bs{Y}}^{h}(\cdot)$ is the orthogonal projection from $\mathcal{T}_{\bs{Y}}\mathcal{M}$ onto $\mathcal{H}_{\bs{Y}}$.

\begin{proposition}\label{proposition:vertical}
The vertical space at $\bs{Y}$ is given by
        \begin{equation}\label{eq:vspace}
        \mathcal{V}_{\bs{Y}}\triangleq \{\bs{Y}\bs{\Omega}:\bs{\Omega}^{\sf{H}}=-\bs{\Omega},\bs{\Omega}\in\mathbb{C}^{r\times r} \}.
        \end{equation}
\end{proposition}
\begin{proof}
     See Appendix \ref{append:vertical}.
\end{proof}
According to the definition, horizontal space should be derived from 
\begin{equation}
        \mathcal{H}_{\bs{Y}} = \{\overline{\bs{\xi}}\in\mathcal{T}_{\bs{Y}}\overline{\mathcal{M}}:\overline{g}_{\bs{Y}}(\overline{\bs{\xi}},\overline{\bs{\zeta}})=0,~\forall \overline{\bs{\zeta}}\in\mathcal{V}_{\bs{Y}}\},
\end{equation}
where $\overline{g}$ is \textit{Riemannian metric} for the abstract manifold $\overline{\mathcal{M}}$. Riemannian metric is the generalization of ``inner product'' in Euclidean space to a manifold. It is a bilinear, symmetric positive-definite operator
\begin{equation}
        \overline{g}:~\mathcal{T}_{\bs{Y}}\overline{\mathcal{M}} \times \mathcal{T}_{\bs{Y}}\overline{\mathcal{M}} \mapsto \mathbb{R}.
\end{equation}
In this paper, we can choose
\begin{equation}\label{eq:metric}
        \overline{g}_{\bs{Y}}(\overline{\bs{\xi}},\overline{\bs{\zeta}}):=\trace(\Re(\overline{\bs{\xi}}^{\sf{H}}\overline{\bs{\zeta}}))=\frac{1}{2}\textrm{Tr}(\overline{\bs{\xi}}^{\sf{H}}\overline{\bs{\zeta}}+\overline{\bs{\zeta}}^{\sf{H}}\overline{\bs{\xi}})
\end{equation}
as a Riemannian metric for the abstract manifold $\overline{\mathcal{M}}$, where $\bs{Y}\in\overline{\mathcal{M}}$ and $\overline{\bs{\xi}},\overline{\bs{\zeta}}\in\mathcal{T}_{\bs{Y}}\overline{\mathcal{M}}$. The manifold $\overline{\mathcal{M}}$ is called a \textit{Riemannian manifold} when its tangent spaces are endowed with a Riemannian metric. From another perspective, $\overline{\mathcal{M}}$ can also be viewed as a K\"{a}hler manifold whose K\"{a}hler form is a real closed (1,1)-form \cite{barth2015compact}. 

Therefore, we can obtain the explicit expressions for the horizontal space and horizontal projection.
\begin{proposition}\label{proposition:horizontal}
        The horizontal space is
        \begin{equation}\label{eq:hspace}
                 \mathcal{H}_{\bs{Y}}=\{\bs{\xi}\in\mathbb{C}^{N\times r}:\bs{\xi}^{\sf{H}}\bs{Y}=\bs{Y}^{\sf{H}}\bs{\xi}\},
         \end{equation}
         and the orthogonal projection onto the horizontal space is 
        \begin{equation}\label{eq:hproj}
        \Pi_{\bs{Y}}^h \bs{\xi}_{\bs{Y}} =  \bs{\xi}_{\bs{Y}}- \bs{Y}\bs{\Omega},
\end{equation}
where $\bs{\Omega}^{\sf{H}}=-\bs{\Omega}\in\mathbb{C}^{r\times r}$ is the solution of Lyapunov equation
\begin{equation}\label{eq:lyap}
        \bs{Y}^{\sf{H}}\bs{Y}\bs{\Omega}+\bs{\Omega}\bs{Y}^{\sf{H}}\bs{Y}=\bs{Y}^{\sf{H}}\bs{\xi}_{\bs{Y}}-\bs{\xi}_{\bs{Y}}^{\sf{H}}\bs{Y}.
\end{equation}
\end{proposition}
\begin{proof}
        See Appendix \ref{append:horizontal}.
\end{proof}

Given the Riemannian metric for the abstract manifold, quotient manifold is naturally endowed with a Riemannian metric
\begin{equation}
        g_{[\bs{Y}]}(\bs{\xi}_{[\bs{Y}]},\bs{\zeta}_{[\bs{Y}]}) : = \overline{g}_{\bs{Y}}(\overline{\bs{\xi}}_{\bs{Y}},\overline{\bs{\zeta}}_{\bs{Y}})
\end{equation}
such that the expression $\overline{g}_{\bs{Y}}(\overline{\bs{\xi}}_{\bs{Y}},\overline{\bs{\zeta}}_{\bs{Y}})$ remains for any elements in the equivalent class $[\bs{Y}]$. 
Hence $\mathcal{M}$ is a \textit{Riemannian quotient manifold} of the abstract manifold $\overline{\mathcal{M}}$ with the Riemannian metric $g$, and the canonical projection $\pi:(\overline{\mathcal{M}},\overline{g})\mapsto (\mathcal{M},g)$ is a \textit{Riemannian submersion} \cite[Sec 3.6.2]{Absil_2009optimizationonManifolds}.

Riemannian optimization generalizes the gradient and Hessian into Riemannian gradient and Riemannian Hessian.

\begin{table*}[ht]
\centering
        \begin{tabular}{c|c r}
        {$\mathop{\textrm{minimize}}_{[{\bs{Y}}]\in\overline{\mathcal{M}}}~ f({\bs{Y}})$} & &  \\\hline
        Computation space  & $\overline{\mathcal{M}}=C_*^{N\times r},\mathcal{M}=\overline{\mathcal{M}}/\sim $ &  \\
        Canonical projection & $\pi:\bs{Y}\mapsto \bs{Y}\bs{Q},{\bs{Q}}\in\mathcal{U}(r)\qquad $ &  (\ref{eq:canonicalprojection}) \\
        Remannian metric & $g_{\bs{Y}}(\bs{\xi}_{\bs{Y}},\bs{\eta}_{\bs{Y}}) = \textrm{Tr}(\bs{\xi}_{\bs{Y}}^{\sf{H}}\bs{\eta}_{\bs{Y}}+\bs{\eta}_{\bs{Y}}^{\sf{H}}\bs{\xi}_{\bs{Y}})\qquad$ & (\ref{eq:metric}) \\
        Vertical space & $\mathcal{V}_{\bs{Y}}=\{\bs{Y}\bs{\Omega}:\bs{\Omega}^{\sf{H}}=-\bs{\Omega},\bs{\Omega}\in\mathbb{C}^{r\times r} \}\qquad$ & (\ref{eq:vspace}) \\
        Horizontal space  & $\mathcal{H}_{\bs{Y}}=\{\bs{\xi}\in\mathbb{C}^{N\times r}:\bs{\xi}^{\sf{H}}\bs{Y}=\bs{Y}^{\sf{H}}\bs{\xi}\}\qquad$ & (\ref{eq:hspace})  \\
        Projection onto horizontal space &  $ \Pi_{\mathcal{M}}^h\bs{\xi}_{\bs{Y}}=\bs{\xi}_{\bs{Y}}- \bs{Y}\bs{\Omega}\qquad$ &(\ref{eq:hproj}) \\
        Remannian gradient & $\textrm{grad}f(\bs{Y})=\sum_{i=1}^{l}(C_i\bs{B}_i+C_i^*\bs{B}_i^{\sf{H}})\bs{Y}\qquad$ & (\ref{eq:finalrgrad}) \\
        Remannian Hessian & 
         $\textrm{Hess}f(\bs{Y})[\bs{\eta}_{\bs{Y}}]\qquad$ & (\ref{eq:finalrhess}) \\
        Retraction & $\mathcal{R}_{\bs{Y}}(\bs{\xi})=\pi(\bs{Y}+\overline{\bs{\xi}})
        \qquad
        $ &(\ref{eq:retraction})
        \end{tabular}
        \caption{Riemannian ingredients}
        \label{tab:riemannian_ingredients}
\end{table*}

\subsubsection{Riemannian Gradient} Riemannian gradient is a necessary ingredient to develop the Riemannian conjugate gradient and Riemannian trust-region algorithm. For quotient manifold, the horizontal representation of Riemannian gradient, denoted by $\textrm{grad}f(\bs{Y})$, arises from
\begin{equation}\label{eq:rgrad}
        \textrm{grad}f(\bs{Y})=\Pi_{\bs{Y}}^h \overline{\textrm{grad}}f(\bs{Y}),
\end{equation}
in which $\overline{\textrm{grad}}f(\bs{Y})$ is the Riemannian gradient in the abstract manifold $\overline{\mathcal{M}}$ at $\bs{Y}$. Note that $\overline{\textrm{grad}}f (\bs{Y})$ is given by
\begin{equation}\label{eq:gradbar}
       \overline{g}_{\bs{Y}}(\overline{\textrm{grad}}f(\bs{Y}),\overline{\bs{\xi}})=\textrm{D}f(\bs{Y})[\overline{\bs{\xi}}],~\forall \overline{\bs{\xi}}\in\mathcal{T}_{\bs{Y}}\overline{\mathcal{M}},
\end{equation}
where $\textrm{D}f(\bs{Y})[\bs{\xi}]:=\lim_{t\rightarrow 0} t^{-1} \left[ f(\bs{Y}+t\bs{\xi})-f(\bs{Y})\right]$ is the directional derivative of $f$, whereas $\overline{\bs{\xi}}$ is the horizontal lift of $\bs{\xi}$. Then we conclude that
\begin{equation}\label{eq:finalrgrad}
        \textrm{grad}f(\bs{Y})=\overline{\textrm{grad}}f(\bs{Y})=\sum_{i=1}^{l}(C_i\bs{B}_i+C_i^*\bs{B}_i^{\sf{H}})\bs{Y},
\end{equation}
in which $C_i=\langle \bs{B}_i,\bs{Y}\bs{Y}^{\sf{H}} \rangle-b_i$. The derivation process is described in detail in Appendix \ref{append:gradhess}.

\subsubsection{Riemannian Hessian}
For the purpose of developing a second-order algorithm, we need to think of the Riemannian Hessian as an linear operator closely connected to the directional derivative of the gradient. \textit{Riemannian connection} defines a ``directional derivative'' on the Riemannian manifold. To be specific, Euclidean directional derivative is a Riemannian connection on $\mathbb{C}^{N\times r}$. Since the quotient manifold $\mathcal{M}$ 
has a Riemannian metric that is invariant along the horizontal space, the Riemannian connection \cite[Proposition 5.3.4]{Absil_2009optimizationonManifolds} can be derived from
\begin{equation}
       \nabla_{\bm{\eta}}\bm{\xi} = \Pi_{\bm{Y}}^{h}(\textrm{D}\overline{\bm{\xi}}[\overline{\bm{\eta}}]), 
\end{equation}
for any $\bm{\eta}\in \mathcal{V}_{\bs{Y}},\; \bm{\xi}\in\mathfrak{X}(\mathcal{M})$ and $\mathfrak{X}(\mathcal{M})$ is the set of smooth vector fields on $\mathcal{M}$. The horizontal representation of Riemannian Hessian operator \cite[Definition 5.5.1]{Absil_2009optimizationonManifolds} is given as
\begin{equation}\label{eq:rhess}
        \textrm{Hess}f(\bs{Y})[\bs{\xi}_{\bs{Y}}]:=\nabla_{\bs{\xi}_{\bs{Y}}}\overline{\textrm{grad}}f.
\end{equation}
Then the Riemannian Hessian is given by
\begin{align}\label{eq:finalrhess}
        \textrm{Hess}f(\bs{Y})[\bs{\eta}_{\bs{Y}}]=\Pi_{\bs{Y}}^{h}\biggl(&\sum_{i=1}^{l}({C_{\bs{\eta}}}_{i}\bs{B}_i\bs{Y} +C_i\bs{B}_i\bs{\eta}_Y \nonumber\\&+ {C_{\bs{\eta}}}_{i}^*\bs{B}_i^{\sf{H}}\bs{Y} +C_i^*\bs{B}_i^{\sf{H}}\bs{\eta}_\bs{Y})\biggr),
\end{align}
where ${C_{\bs{\eta}}}_{i}=\langle \bs{B}_i,\bs{Y}\bs{\eta}_\bs{Y}^{\sf{H}} +\bs{\eta}_\bs{Y}\bs{Y}^{\sf{H}}\rangle$.
We relegate the derivation details of this expression to Appendix \ref{append:gradhess}.

\subsection{Riemannian Optimization for Fixed-Rank Problem} 
Riemannian optimization generalizes the optimization algorithms in Euclidean space to a manifold. Similarly, we need to compute search directions in the tangent space and appropriate stepsizes. To ensure each iteration is always on the given manifold, \textit{retraction} \cite[Sec 4.1]{Absil_2009optimizationonManifolds} is defined as a pull-back from the tangent space onto the manifold. To be specific, the updating formula in the $i$-th iteration is given by
\begin{equation}
        \bs{Y}_{k+1} = \mathcal{R}_{\bs{Y}_k}(\alpha_k\bs{\eta}_{k}),
\end{equation}
where $\alpha_k>0$ is the step size, $\bs{\eta}_k \in \mathcal{T}_{\bs{Y}_k}\mathcal{M}$ is the search direction, and $\mathcal{R}$ denotes retraction operation which maps an element from the set of all tangent spaces $\mathcal{T}\mathcal{M}=\cup_{\bs{Y}\in\mathcal{M}}\mathcal{T}_{\bs{Y}}\mathcal{M}$ to $\mathcal{M}$. The retraction operation is shown in Fig. \ref{fig:retraction}.
\begin{proposition}
Choices of $\overline{\mathcal{R}}$ and $\mathcal{R}$
\begin{equation}\label{eq:retraction}
        \overline{\mathcal{R}}_{\bs{Y}}(\overline{\bs{\xi}}): = \bs{Y}+\overline{\bs{\xi}},~~
        \mathcal{R}_{\bs{Y}}(\bs{\xi}):=\pi(\overline{\mathcal{R}}_{\bs{Y}}(\overline{\bs{\xi}}))
\end{equation}
define retractions on $\overline{\mathcal{M}}$ and $\mathcal{M}$, respectively.
\end{proposition}
\begin{proof}
Since $\overline{\mathcal{M}}=\mathbb{C}_*^{N\times r}$ is an embedded manifold and also an open submanifold of $\mathcal{E}=\mathbb{C}^{N\times r}$, following \cite[Sec 4.1.1]{Absil_2009optimizationonManifolds} we can choose the identity mapping
\begin{equation}
        \phi(\bs{F}) = \bs{F}
\end{equation}
as a diffeomorphism so that $\phi:\overline{\mathcal{M}} \times \overline{\mathcal{N}}\rightarrow  \mathcal{E}_*,\overline{\mathcal{N}}=\emptyset$ and $\textrm{dim}(\overline{\mathcal{M}})+\textrm{dim}(\overline{\mathcal{N}})=\textrm{dim}(\mathcal{E})$. Therefore, we conclude that
\begin{equation}
        \overline{\mathcal{R}}_{\bs{Y}}(\overline{\bs{\xi}}): =\bs{Y}+\overline{\bs{\xi}}
\end{equation}
defines a retraction on $\overline{\mathcal{M}}$. Adding with that equivalent classes are orbits of the Lie group $\mathcal{U}_r$ which acts linearly \cite[Sec 4.1.2]{Absil_2009optimizationonManifolds} on the abstract manifold $\overline{\mathcal{M}}$,
\begin{equation}
        \mathcal{R}_{\bs{Y}}(\bs{\xi}):=\pi(\overline{\mathcal{R}}_{\bs{Y}}(\overline{\bs{\xi}}))
\end{equation}
defines a retraction on $\mathcal{M}$.
\end{proof}
In this subsection, we will introduce Riemannian conjugate gradient (RCG) method and Riemannian trust-region (RTR) method.
\subsubsection{Riemannian Conjugate Gradient Method}\label{sec:RCG}
When the search direction is chosen as the negative Riemannian gradient and the step size is determined by backtracking line search following the Armijo rule \cite[4.6.3]{Absil_2009optimizationonManifolds}, we have the Riemannian gradient descent algorithm. Riemannian conjugate gradient method can be expressed as
\begin{equation}
        \bs{\eta}_{k+1} = -\textrm{grad}f_k+\beta_k \mathfrak{T}_{\alpha_k\bs{\eta}_{k}}(\bs{\eta}_{k}),
\end{equation}
where $\mathfrak{T}_{\bs{\eta}_{\bs{Y}}}(\bs{\xi}_{\bs{Y}})$ is the \textit{vector transport} operator so that $\bs{\xi}_{\bs{Y}}$ is transported from $\mathcal{T}_{\bs{Y}}\mathcal{M}$ to $\mathcal{T}_{\mathcal{R}_{\bs{Y}}(\bs{\eta}_{\bs{Y}})}\mathcal{M}$ for $\bs{\xi}_{\bs{Y}}\in\mathcal{T}_{\bs{Y}}\mathcal{M}$. This is shown in Fig. \ref{fig:transport}.
\begin{figure}[t]
        \centering
        \subfloat[Retraction]{\includegraphics[width=2.2in]{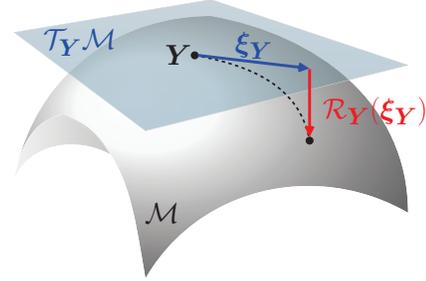}\label{fig:retraction}}\hfil
        \subfloat[Vector Transport]{\includegraphics[width=2.2in]{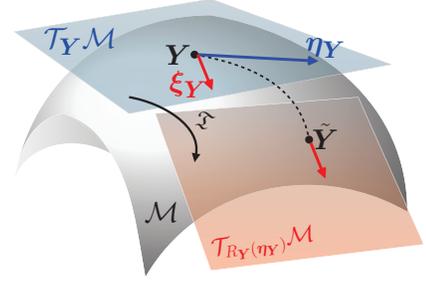}\label{fig:transport}}
        \caption{Riemannian retraction and vector transport operation.}
\end{figure}

A vector transport is defined by
\begin{equation}\label{eq:transport}
        \mathfrak{T}_{\bs{\eta}_{\bs{Y}}}(\bs{\xi}_{\bs{Y}}):=D\mathcal{R}_{\bs{Y}}(\bs{\eta}_{\bs{Y}})[\bs{\xi}_{\bs{Y}}]=\bs{\xi}_{\bs{Y}}.
\end{equation}
Among many good choices for $\beta_k$, we choose 
\begin{equation}
        \beta_{k} =\frac{g_{\bs{Y}_k}(\textrm{grad}f_k,\textrm{grad}f_k-\mathfrak{T}_{\alpha_{k-1}\bs{\eta}_{k-1}}(\textrm{grad}f_{k-1}))}{g_{\bs{Y}_k}(\bs{\eta}_{k-1},\textrm{grad}f_k-\mathfrak{T}_{\alpha_{k-1}\bs{\eta}_{k-1}}(\textrm{grad}f_{k-1}))}, \nonumber
\end{equation}
which is a generalized version of Hestenes-Stiefel \cite{hestenes1952methods}. 

\subsubsection{Riemannian Trust-Region Algorithm}\label{sec:RTR}
When the search direction is chosen by solving the local second-order approximation of $f(\bs{Y})$, it results in the Riemannian trust-region algorithm. We will find the updating vector $\bs{\eta}$ by solving the trust-region subproblem
\begin{eqnarray}
        \mathop{\textrm{minimize}}_{\bs{\eta}\in\mathcal{T}_{\bs{Y}_k}\mathcal{M}}&&m_k(\bs{\eta})=f_k+g_{\bs{Y}_k}(\textrm{grad}f_k,\bs{\eta}) +\frac{1}{2}g_{\bs{Y}_k}(\textrm{Hess}f_k[\bs{\eta}],\bs{\eta}) \nonumber\\
        \textrm{subject to}&&\|\bs{\eta}\|_g\leq \Delta_k, \label{prob:tr_sub}
\end{eqnarray}
where $\Delta_k>0$ is the radius of the trust region and $\|\bs{\eta}\|_g=\sqrt{g_{\bs{Y}_k}(\bs{\eta},\bs{\eta})}$. Note that the solution $\bs{\eta}$ becomes a candidate for updating. This is because we will select a proper $\Delta_k$, find the corresponding solution $\bs{\eta}_k$ of the trust-region subproblem and then update $\bs{Y}$ through
\begin{equation}
     \bs{Y}_{k+1} = \mathcal{R}_{\bs{Y}_k}(\bs{\eta}_{k}).
\end{equation}
The criterion for choosing $\Delta_k$ is based on evaluating
\begin{equation}
        \rho_k = \frac{f_k-f(\mathcal{R}_{\bs{Y}_k}(\bs{\eta}_k))}{m_k(\bs{0})-m_k(\bs{\eta}_k)}.
 \end{equation} 
 When $\rho_k$ is very small, the radius of trust region $\Delta_k$ should be reduced because in this case the second-order approximation is too inaccurate. If $\rho_k$ is not very small, we shall accept $\Delta_k$ and $\bs{\eta}_k$ and reduce the trust region. If $\rho_k$ is close to 1, it m
 means that the second-order approximation models original objective function well. Hence, we can accept this step and expand
the trust region. Likewise, if $\rho_k\ll 1$, we should also reject this step and increase $\Delta_k$. The trust-region subproblem can be solved by the truncated conjugate gradient \cite[Sec 7.3.2]{Absil_2009optimizationonManifolds} algorithm (see Algorithm \ref{algorithm:tcg}).
 \SetNlSty{textbf}{}{:}
\IncMargin{1em}
\begin{algorithm}[tb]
\textbf{Input:} $\mathcal{B},\bm{b},\bs{Y}_k,\Delta_k$. Parameters $\kappa,\theta>0$\\
 Initialize: $\bm{\eta}^{0},\bs{r}_0=\textrm{grad}f_k,\bs{\delta}_0=-\bs{r}_0$\\
 \While{$\|\bs{r}_{j+1}\|_{g}>\|\bs{r}_{0}\|_{g}\min(\|\bs{r}_{0}\|_{g}^{\theta},\kappa)$}{
\If{$g_{\bs{Y}_k}(\bs{\delta}_j,\textrm{Hess}f_k[\bs{\delta}_j])\leq 0 $}
{Compute $\tau=\arg\min m_k(\bs{\eta}_k)$ where $\bs{\eta}_k=\bs{\eta}^j+\tau\bs{\delta}_j$ and $\|\bs{\eta}_k\|_g= \Delta_k$,\\
\Return $\bs{\eta}$}
Set $\bs{\eta}^{j+1}=\bs{\eta}^j+\alpha_j\bs{\delta}_j$ where $\alpha_j=\|\bs{r}_j\|_g^2/g_{\bs{Y}_k}(\bs{\delta}_j,\textrm{Hess}f_k[\bs{\delta}_j])$\\
\If{$\|\bs{\eta}^{j+1}\|_g\geq {\Delta}_k $}
{Set $\tau$ as the solution to $\|\bs{\eta}_k\|_g= \Delta_k$ where $\bs{\eta}_k=\bs{\eta}^j+\tau\bs{\delta}_j$,\\
\Return $\bs{\eta}_k$}
Set $\bs{r}_{j+1}=\bs{r}_j+\alpha_j\textrm{Hess}f_k[\bs{\delta}_j]$,\\
Set $\beta_{j+1}=\|\bs{r}_{j+1}\|_{g}^2/\|\bs{r}_{j}\|_{g}^2$,\\
Set $\bs{\delta}_{j+1}=-\bs{r}_{j+1}+\beta_{j+1}\bs{\delta}_j$,\\
$j\leftarrow j+1$
}
 \textbf{Output:} $\bs{\eta}=\bs{\eta}_k$.
 \caption{Truncated Conjugate Gradient Algorithm for (\ref{prob:tr_sub})}
 \label{algorithm:tcg}
\end{algorithm}

Riemannian trust-region algorithm harnesses the second-order information of the problem. It admits a superlinear \cite[Theorem 7.4.11]{Absil_2009optimizationonManifolds} convergence rate locally and is robust to initial points. Since the objective function $f$ is exactly a quadratic function which satisfies the Lipschitz gradient condition and other assumptions in \cite{boumal2016global}, we can always find an approximate second-order critical points by the Riemannian trust-region algorithm.

\subsection{Computational Complexity Analysis}
Riemannian conjugate gradient and Riemannian trust-region algorithm involve computing optimization ingredients at each iteration, for which we show their computational complexity.
\begin{itemize}
        \item Evaluate the objective value $f(\bs{Y})$. Since $\mathcal{B}(\cdot)$ involves a series of sparse matrix multiplication, we can compute it efficiently and the complexity of computing $f(\bs{Y})$ is $\mathcal{O}(mnl)$.
        \item Compute the Riemannian gradient $\textrm{grad}f$ (\ref{eq:finalrgrad}). This includes computing matrix multiplication in $\mathcal{O}(mnl)$ and horizontal projection $\Pi_{\bs{Y}}^h$ (\ref{eq:hproj}). Since complexity of solving the Lyapunov equation (\ref{eq:lyap}) is $\mathcal{O}(r^3+(m+n)r^2)$,  the overall complexity is $\mathcal{O}(mnl+(m+n)r^2+r^3)$.
        \item Compute the Riemannian Hessian $\textrm{Hess}f$ (\ref{eq:finalrhess}). Its cost is also $\mathcal{O}(mnl+(m+n)r^2+r^3)$.
        \item Computing the Riemannian metric $\overline{g}$ (\ref{eq:metric}). This complexity is dominant by matrix multiplications, which is $\mathcal{O}((m+n)r^2)$.
        \item Computational complexity of retraction $\overline{\mathcal{R}}$ (\ref{eq:retraction}) is $\mathcal{O}((m+n)r)$ and vector transport $\mathfrak{T}$ (\ref{eq:transport}) is insignificant.
\end{itemize}
From the above results, we conclude that the computational complexity of Riemannian conjugate gradient algorithm for each iteration is $\mathcal{O}(mnl+(m+n)r^2+r^3)$. And each iteration of truncated conjugate gradient algorithm also involves computation with complexity $\mathcal{O}(mnl+(m+n)r^2+r^3)$.

\section{Simulations} \label{sec:simulation}
This section presents numerical experiments to demonstrate the efficacy of the generalized low-rank optimization approach for topological cooperation via the newly presented Riemannian optimization algorithms. We will investigate the performance of different algorithms from the perspective of convergence rate and achievable DoF. We evaluate our model 
in different settings and demonstrate that the generalized low-rank approach can effectively enable transmitter cooperation based 
only on network topology information.

Our simulations compare the following matrix-factorization-based algorithms 
 for solving the generalized low-rank optimization problem $\mathscr{P}$:
\begin{itemize}
        \item ``AltMin'' \cite{yi2016topological}: Alternating minimization algorithm (\ref{GLRM_fixedrank_altmin1}) (\ref{GLRM_fixedrank_altmin2}) is adopted in \cite{yi2016topological} for topological transmitter cooperation problem by alternatively updating factors. For fixed $r$, (\ref{GLRM_fixedrank_altmin1}) (\ref{GLRM_fixedrank_altmin2}) are solved with gradient descent followed by backtracking line search.
        \item ``RCG'': The Riemannian conjugate gradient method is developed in Sec \ref{sec:RCG}. We implement this algorithm with \textit{Manopt} \cite{boumal2014manopt} software package.
        \item ``RTR'': The Riemannian trust-region method is developed in Sec \ref{sec:RTR} and also implemented with \textit{Manopt}. 
\end{itemize}
All algorithm are adopted with random initialization strategy for each rank $r$, and we find the minimal $r$ by increasing $r$ from 1 to $N$ until ${m}^{-0.5}\cdot\|\mathcal{A}(\bs{X})-\bs{b}\|<10^{-3}$. In our numerical experiments, the network topology and shared messages at each transmitter are generated uniformly at random with probabilities
\begin{equation}
        \textrm{Prob}((k,j)\in\mathcal{E})=\left\{\begin{aligned}
                & p, && j\ne k \\
                & 1, && j=k
        \end{aligned}\right.,
\end{equation}
and
\begin{equation}
\textrm{Prob}(j\in\mathcal{S}_k)=\left\{\begin{aligned}
                & q, && j\ne k \\
                & 1, && j=k
        \end{aligned}\right.,
\end{equation}
respectively.

\begin{figure}[h]
        \centering
        \subfloat[]{\includegraphics[width=\columnwidth]{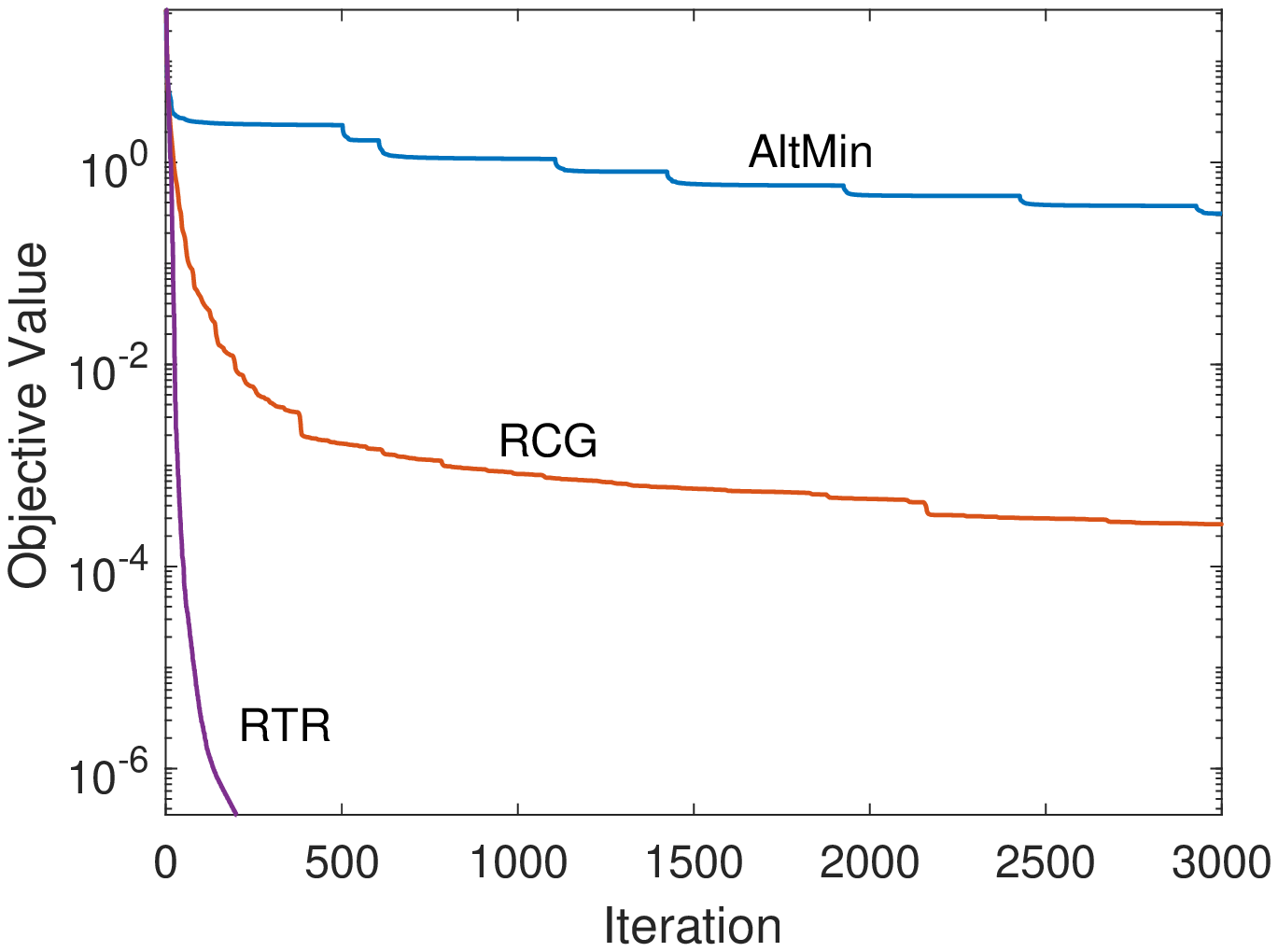}}\hfil
        \subfloat[]{\includegraphics[width=\columnwidth]{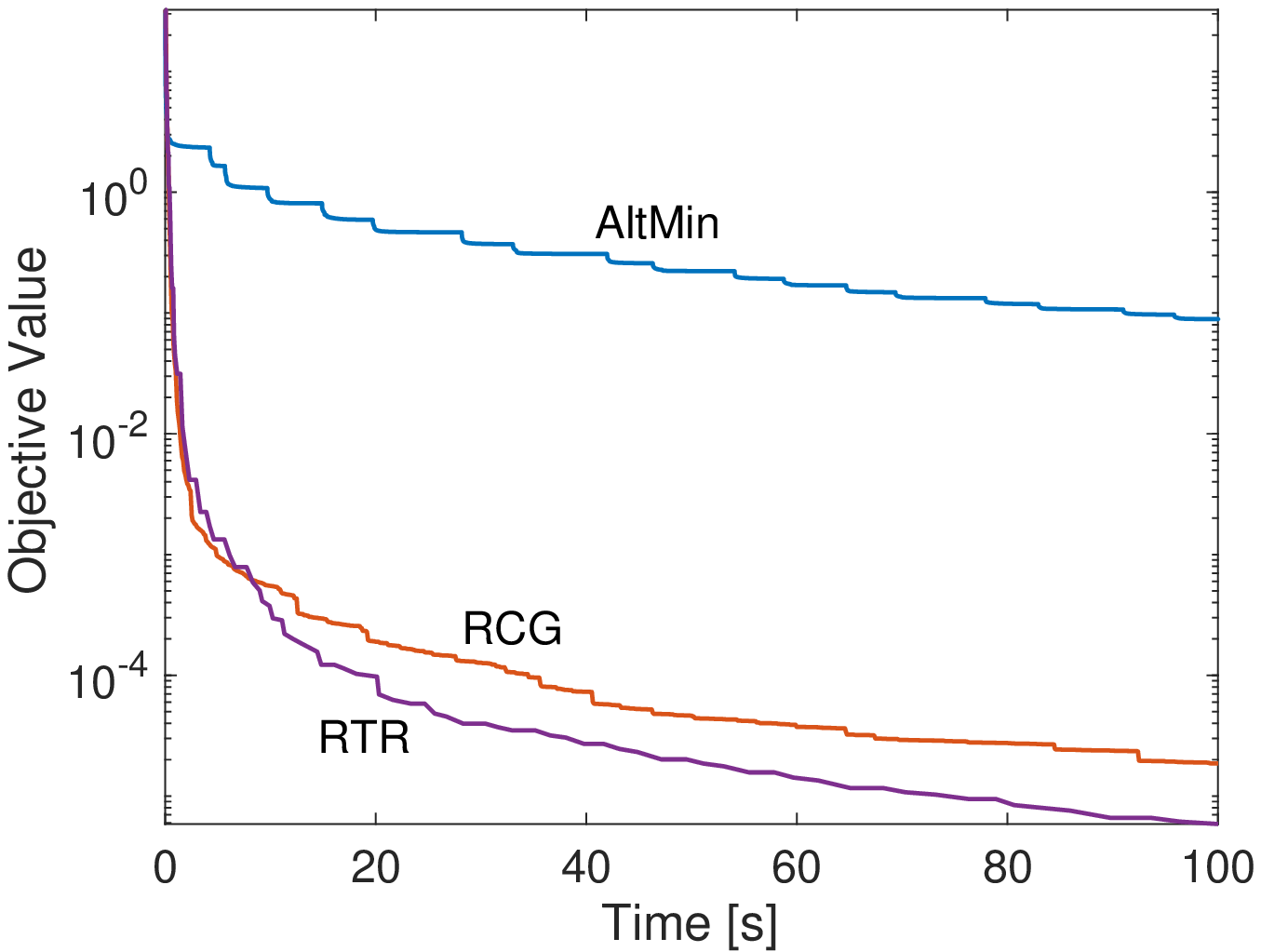}}
        \caption{Convergence rate and computing time of all algorithms with full transmitter cooperation.}
        \label{fig:convergence}
\end{figure}
\subsection{Convergence Rate}
Consider a partially connected $20$-user interference channel with full transmitter cooperation. The network topology is generated randomly and each link is connected with probability $p=0.3$. Each message is split into $3$ data streams. In this simulation, $K=20, d_1=\cdots=d_K=3$, $p=0.3$ and $r=12$. Fig. \ref{fig:convergence} shows 
the convergence behaviors of all 3 algorithms in terms of 
iterations and time. 
The results indicate that the proposed RTR algorithm 
exhibits a superlinear convergence rate, and the computing 
rate is comparable with first order RCG algorithm. 
In addition, the proposed RTR can yield a more accurate 
solution with the second-order stationary point when
compared against first order algorithms that guarantee convergence
only to first-order stationary points. 
The overall test results show that the proposed RTR and RCG 
algorithms are much more efficient than other contemporary algorithms
in terms of convergence rates and solution performance.  

To further show that the interferences are nulled, we choose the interference leakage as the metric and plot it in Fig. \ref{fig:leakage} for the same setting of Fig. \ref{fig:convergence}. The interference leakage cost is given by
\begin{equation}
    IL = \sum_{i\ne k}\sum_{j:(k,j)\in\mathcal{E},i\in\mathcal{S}_j}\!\!\!\! \!\!\! \|h_{kj}\bs{U}_{k}^{\sf{H}}\bs{V}_{ji}\Big(h_{kj}\bs{U}_{k}^{\sf{H}}\bs{V}_{ji}\Big)^{\sf{H}}\|_F^2,
\end{equation}
where the channel coefficients follow standard complex Gaussian distribution. Fig. \ref{fig:leakage} demonstrates that there is a rapid decline of interference leakage as the objective value decreases with the proposed Riemannian optimization algorithms.
\begin{figure}[h]
        \centering
        \subfloat[]{\includegraphics[width=\columnwidth]{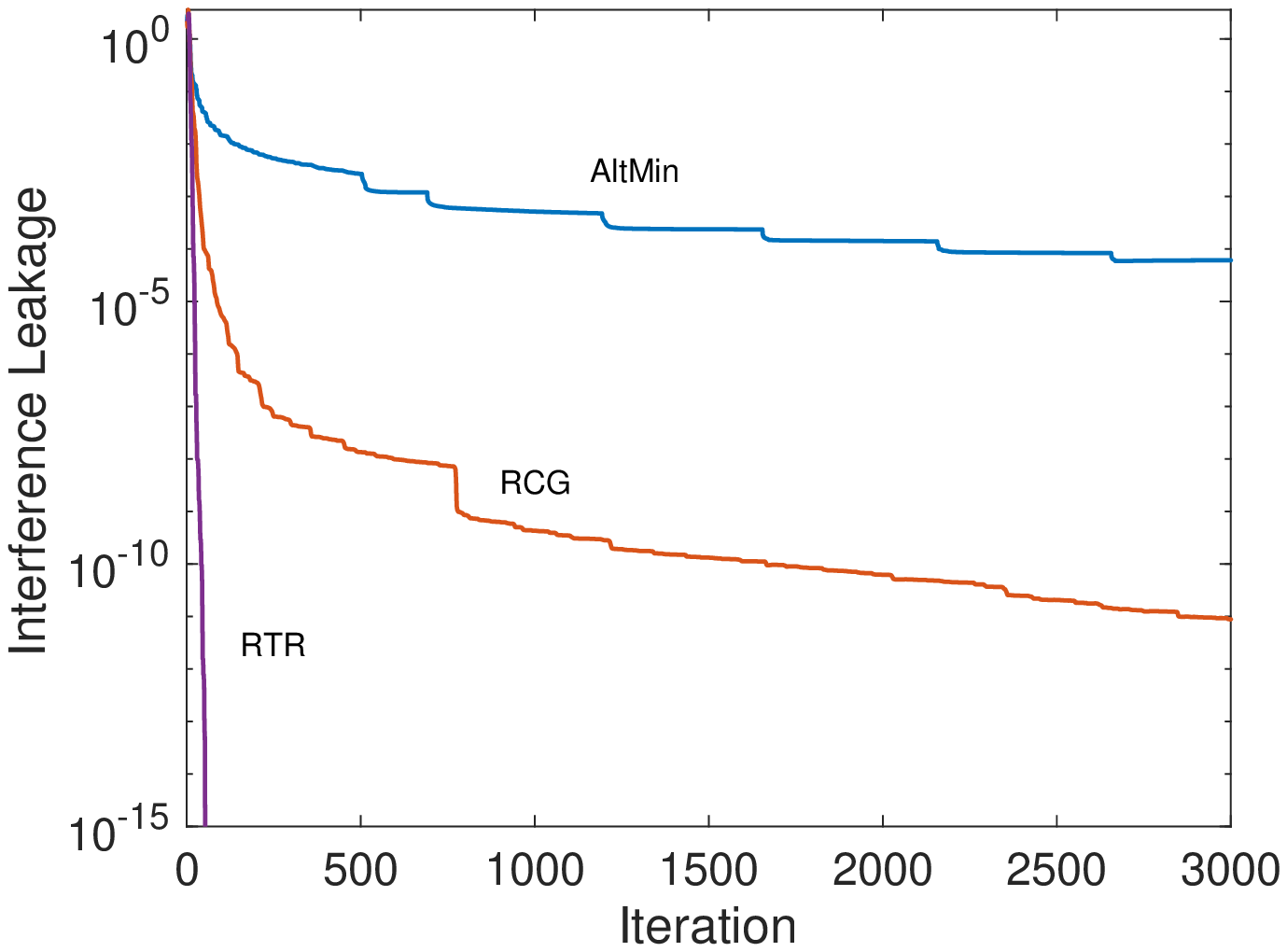}}\hfil
        \subfloat[]{\includegraphics[width=\columnwidth]{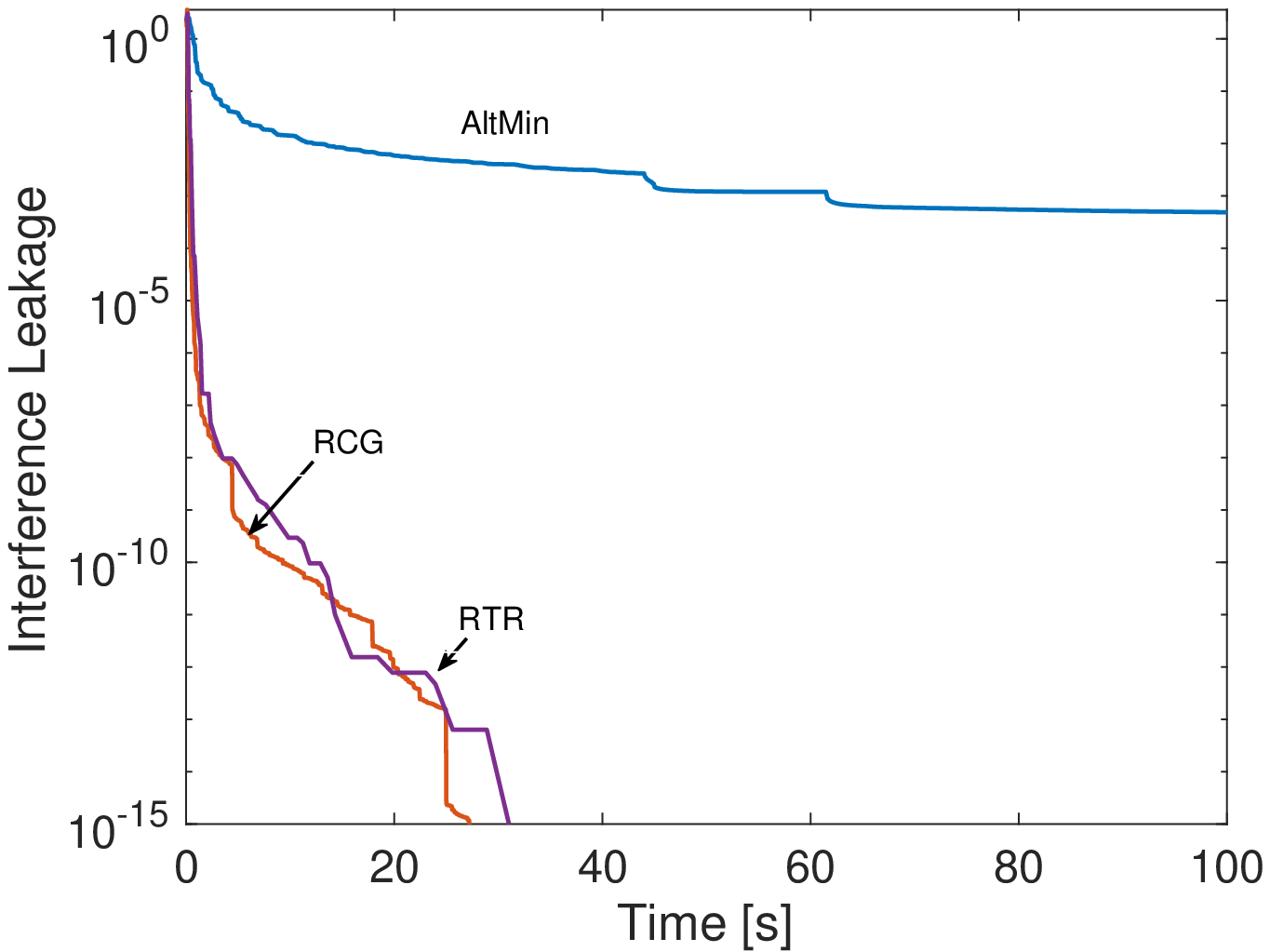}}
        \caption{Convergence of interference leakage for all algorithms with full transmitter cooperation.}
        \label{fig:leakage}
\end{figure}

\subsection{DoF over Network Topologies}
Consider a partially connected $20$-user interference channel without message splitting ($d_k=1$). The network topologies are generated randomly with different $p$. Fig. \ref{fig:connectedlinks} demonstrates the DoF over $p$ with full transmitter cooperation. Each DoF result
is averaged over 100 times. This result shows that, among the 3 
solutions,  the 
proposed RTR algorithm achieves the best performance with 
second-order stationary points. 
The Riemannian algorithms RTR and RCG significantly outperform the alternating minimization algorithm owing to their good convergence guarantee.
\begin{figure}[h]
        \centering
        \includegraphics[width=\columnwidth]{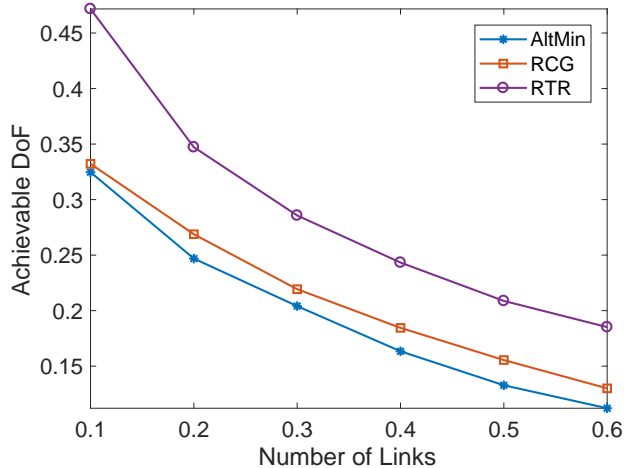}
        \caption{DoF over the number of connected links.}
        \label{fig:connectedlinks}
\end{figure}

To further justify the effectiveness of the Riemannian optimization framework, we check the recovered DoF returned by the proposed RTR algorithm for all specific network topologies $\mathcal{E}$ and specific message sharing pattern $\{\mathcal{S}_k:k\in[K]\}$ provided in \cite{yi2015topological}. Specifically, transmitter cooperation improves the symmetric DoF from $1/3$ to $2/5$ for Example 1 in Fig. 1(a), from $2/5$ to $1/2$ for Example 4 in Fig. 3(a), and from $1/3$ to $2/5$ for Example 7 in Fig. 6(a), compared with the cases without cooperation. And the optimal symmetric DoF is $1/2$ for Example 6 in Fig. 5(a) with transmitter cooperation. All these optimal symmetric DoF results can be achieved by the proposed RTR algorithm numerically. However, theoretically identifying the network topologies and the message sharing patterns for which the Riemannian trust region algorithm can provide optimal symmetric DoFs is still a challenging open problem.

\subsection{Transmitter Cooperation Gains}
We investigate the achievable DoFs in partially connected $20$-user interference channels. We randomly generate the network topologies with $p=0.2$ and simulate different algorithms under different transmitter cooperation level $q$ with single data stream. For each cooperation level, we take average over 500 channel realizations. 
Fig.~\ref{fig:cooperationlevel} shows that the 
second-order algorithm RTR can achieve 
the highest DoF among all algorithms. 
Comparing the first-order algorithms, the proposed RCG  
outperforms AltMin. With the high convergence rate and second-order stationary points solutions of RTR, the gap between RTR algorithm and other algorithms grows with $q$, which indicates 
that the proposed RTR algorithm is capable 
of fully leveraging the benefits of transmitter cooperation.
\begin{figure}[h]
        \centering
        \subfloat[]{\includegraphics[width=\columnwidth]{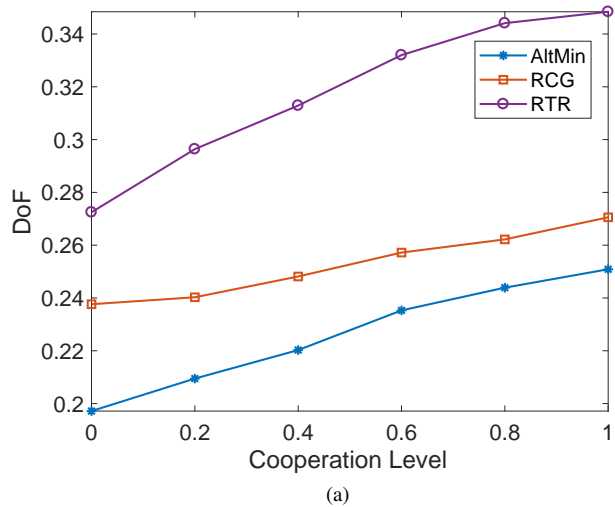}\label{fig:cooperationlevel}}\hfil
        \subfloat[]{\includegraphics[width=\columnwidth]{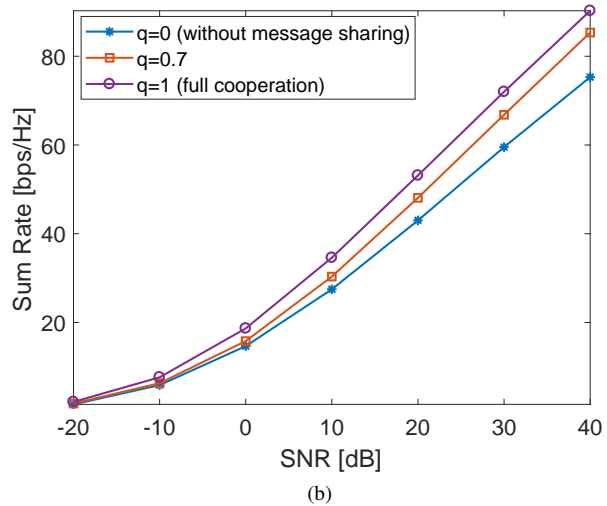}\label{fig:sumrate}}
        \caption{Benefits of transmitter cooperation: (a) DoF over different transmitters cooperation levels $q$. (b) Sum-rate over the transmit power.}
\end{figure}


To further illustrate the transmitter cooperation benefit, we evaluate the achievable sum-rate using the proposed RTR algorithm in 
Fig.~\ref{fig:cooperationlevel}. In this single data stream test
setting, $\bs{U}_k,\bs{V}_{ji}$ degenerates to vectors
 $\bs{u}_k,\bs{v}_{ji}$. Assume that each single data stream 
symbol $s_i$ has unit power, i.e., $\mathbb{E}(|s_i|^2)=1$. 
Suppose the noise is i.i.d. Gaussian, i.e., $\bs{\Sigma}_k = \sigma^2\bs{I}_r$, with $\sigma^2$ at $-120$dB. The distance $d_{ij}$ between each connected transmitter-receiver pair $(j,i)$ is uniformly distributed in $[0.1, 0.2]$ km. The fading channel model is given as
\begin{equation}
        h_{ij} = 10^{-L(d_{ij})/20}c_{ij}, (i,j)\in\mathcal{E},
\end{equation}
where the pass loss is given by $L(d_{ij}) = 128.1+37.6\log_{10}d_{ij}$ and the small scale fading coefficient is given by $c_{ij}\sim\mathcal{CN}(0,1)$. Then the sum rate per channel use is given by
\begin{equation}
        C_{\text{sum}}=\frac{1}{r}\sum_{k=1}^{K}C_k = \frac{1}{r}\sum_{k=1}^{K}\log(1+\textrm{SINR}_k),
\end{equation}
where
\begin{equation*}
        \textrm{SINR}_k = \frac{|\sum_{(k,j)\in\mathcal{E},k\in\mathcal{S}_j}h_{kj}\bs{u}_k^{\sf{H}}\bs{v}_{jk}|^2}{\sum_{i\ne k}|\sum_{(k,j)\in\mathcal{E},i\in\mathcal{S}_j}h_{kj}\bs{u}_k^{\sf{H}}\bs{v}_{ji}|^2+\|\bs{u}_k\|_2^2\sigma^2}.
\end{equation*}
Based on the singular value decomposition (SVD) $\bs{X}^*=\bs{U}\bs{\Sigma}\bs{V}^{\sf{H}}$, the transmit beamformer is simply chosen as $\bs{U}\bs{\Sigma}^{\frac{1}{2}}$, and receive beamformer is given by $\bs{V}\bs{\Sigma}^{\frac{1}{2}}$ with power normalization. Fig. \ref{fig:sumrate} shows the achievable sum-rate over different transmit power $P$. Each point is averaged 
across 100 channel realizations with random 
$\mathcal{E}$ ($p=0.2$). The result also demonstrates that the 
proposed RTR algorithm is capable of achieving high data 
rates by leveraging transmitter cooperation.

In summary, our numerical experiments demonstrate that the
proposed Riemannian trust-region algorithm is capable of obtaining high-precision solutions with second-order stationary points, 
leading to high achievable DoFs and data rates. 
Furthermore, the computation time of RTR algorithm in medium scale problems is comparable with first order algorithms, which demonstrates that the proposed RTR algorithm is a powerful algorithm capable of
harnessing the benefit of topological cooperation in
problems involving medium network sizes.

\section{Conclusions} \label{sec:conclusion}

This work investigates the opportunities of transmitter
cooperation based only on topological information
with message sharing. Our contributions include the
derivation of a generalized topological interference alignment condition, followed by the development of
a low-rank matrix optimization approach to maximize the achievable DoFs. 
To solve the resulting  generalized low-rank optimization problem 
which is nonconvex in complex field, we developed Riemannian
optimization algorithms by exploiting the complex non-compact Stiefel manifold for fixed-rank matrices in complex field. In particular, we 
adopted
the semidefinite lifting technique and Burer-Monteiro factorization approach. Our experiments demonstrated that the proposed Riemannian algorithms considerably outperformed the 
alternating minimization algorithm. Additionally, 
the proposed Riemannian trust-region algorithm achieves high DoFs with high-precision 
second-order stationary point solutions, with computation
complexity comparable with the first-order Riemannian 
conjugate gradient algorithm.

\appendices
\section{Proof of Proposition \ref{proposition:nuc}}\label{append:nuc}
For simplicity, we only give the proof of the single data stream case, while it can be readily extended to general multiple data stream cases. In this case, $\bs{X}$ is given by
\begin{equation}
        \bs{X} = \begin{bmatrix}
                \bs{u}_1^{\sf{H}}\bs{v}_{11} & \cdots & \bs{u}_1^{\sf{H}}\bs{v}_{KK} \\
                \vdots & \ddots & \vdots \\
                \bs{u}_K^{\sf{H}}\bs{v}_{11} & \cdots & \bs{u}_K^{\sf{H}}\bs{v}_{KK}
        \end{bmatrix}\in\mathbb{C}^{K\times K^2}.
\end{equation}

Let $\bs{x}_i$ denote the transpose of the $i$-th row of matrix $\bs{X}$. Problem (\ref{prob:nuc}) can be rewritten as
\begin{eqnarray}\label{prob:re_nuc}
\mathop{\textrm{minimize}}_{{\bs{X}}\in\mathbb{C}^{m\times n}}&& \|{\bs{X}}\|_* \nonumber\\
\textrm{subject to}&& \bm{1}^{\sf{H}}\bs{x}_k^{\mathcal{D}_k}=1,~\forall k\in [K]  \nonumber\\
&&\bs{x}_{k}^{\mathcal{G}_k} = \bs{0},~\forall k\in [K], 
\end{eqnarray}
in which $\bs{x}_k^{\mathcal{D}_k},\bs{x}_k^{\mathcal{G}_k}$ are vectors whose elements are sampled from $\bs{x}_k$, and $\mathcal{D}_k,\mathcal{G}_k$ are the index sets of sampling. $\mathcal{D}_k$ and $\mathcal{G}_k$ are given by
\begin{equation}
        \mathcal{D}_k = \{(j-1)*K+k:(k,j)\in\mathcal{E},k\in\mathcal{S}_j\}
\end{equation}
and
\begin{equation}
        \mathcal{G}_k = \{(j-1)*K+i:i\ne k, (k,j)\in\mathcal{E},i\in\mathcal{S}_j\},
\end{equation}
respectively.  
\begin{lemma}\label{prop:optimal}
        The optimal solution of (\ref{prob:re_nuc}), denoted by $\bs{X}^\star$, is given by
\begin{equation}
        {\bs{x}^\star}_k^{\mathcal{D}_k} = \frac{1}{|\mathcal{D}_k|}{\bs{1}},
\end{equation}
where $\bs{1}$ denotes the vector of all ones. The remaining entries of $\bs{X}^\star$ are all zeros.
\end{lemma}

\begin{proof}
        Let $g(\bs{X})=\|\bs{X}\|_*=\textrm{Tr}(\sqrt{\bs{X}\bs{X}^{\sf{H}}})$. To proof Lemma \ref{prop:optimal}, it is equivalent to prove that $t=0$ is a minimum of the convex function 
        \begin{equation}
                h(t) = \|\bs{X}^\star+t\bs{X}\|_*,\mathcal{A}(\bs{X}^\star+t\bs{X})=\bs{b},t\in\mathbb{R}.
        \end{equation}
        This can be deduced from the fact that any feasible point can be expressed as $\bs{X}^\star+t\bs{X}$, and if $t=0$ is a minimum of $h(t)$, then $g(\bs{X}^\star)\leq g(\bs{X})$ always holds.

        Based on the structure of $\mathcal{D}_k$ and $\mathcal{G}_k$, we have
        \begin{equation}\label{eq:proof2}
                \bs{X}^\star{\bs{X}^\star}^{\sf{H}}=\begin{bmatrix}
                        \frac{1}{|\mathcal{D}_1|} & 0 & \cdots & 0 \\
                        0 & \frac{1}{|\mathcal{D}_2|} & \cdots & 0 \\
                        \vdots & \vdots & \ddots & \vdots \\
                        0 & 0 & \cdots & \frac{1}{|\mathcal{D}_K|}
                \end{bmatrix}.
        \end{equation}
        Then $\bs{X}^\star$ is a full rank matrix, and we can find $|t|\leq \epsilon$ such that $\bs{X}^\star+t\bs{X}$ is invertible. Therefore, the derivative of $h$ is given by
        \begin{align}
                h^\prime(t)=  \frac{1}{2}\big\langle &\left((\bs{X}^\star+t\bs{X})(\bs{X}^\star+t\bs{X})^{\sf{H}}\right)^{-\frac{1}{2}}, \nonumber\\&\bs{X}{\bs{X}^\star}^{\sf{H}}+\bs{X}^\star\bs{X}^{\sf{H}}+2t\bs{X}\bs{X}^{\sf{H}} \big\rangle.
        \end{align}
        Then we have
        \begin{equation}
                h^\prime(0)=\frac{1}{2}\big\langle \left(\bs{X}^\star{\bs{X}^\star}^{\sf{H}}\right)^{-\frac{1}{2}},\bs{X}{\bs{X}^\star}^{\sf{H}}+\bs{X}^\star\bs{X}^{\sf{H}} \big\rangle.
        \end{equation}
        Since $ \mathcal{A}(\bs{X}^\star+t\bs{X})=\bs{b}$, i.e.,
        \begin{equation}
                \bm{1}^{\sf{H}}({\bs{x}^\star}_k^{\mathcal{D}_k}+t\bs{x}_k^{\mathcal{D}_k}) = 1
        \end{equation}
        then we have $\langle{\bs{x}^\star}_k^{\mathcal{D}_k},\bs{x}_k^{\mathcal{D}_k}\rangle = 0$. Therefore,
        \begin{equation}\label{eq:proof1}
                \textrm{diag}(\bs{X}{\bs{X}^\star}^{\sf{H}}) = \textrm{diag}(\bs{X}^\star\bs{X}^{\sf{H}}) = \bs{0}. 
        \end{equation}
        From (\ref{eq:proof1}) (\ref{eq:proof2}), we can deduce that $h^\prime(0)=0$, and thus $t=0$ is a minimum of $h(t)$.
\end{proof}
From Lemma \ref{prop:optimal} we know that the optimal solution $\bs{X}^{\star}$ of (\ref{prob:nuc}) is full rank. So nuclear norm relaxation approach always fails.

\section{Proof of Proposition \ref{proposition:vertical}}\label{append:vertical}
The elements in the vertical space $\mathcal{V}_{\bs{Y}}$ must be tangential to the equivalent class $[\bs{Y}]=\{\bs{Y}\bs{Q}:\bs{Q}^{\sf{H}}\bs{Q}=\bs{I}\}$. Let $\bs{Y}(t)=\bs{Y}_0\bs{Q}(t)$ be a curve in $[\bs{Y}_0]$ through $\bs{Y}_0$ at $t=0$, i.e., $\bs{Q}(0)=\bs{I}$. Then we have
\begin{equation}\label{eq:tangentcurve}
        \bs{Y}(t)\bs{Y}(t)^{\sf{H}}=\bs{Y}_0\bs{Q}(t)\bs{Q}(t)^{\sf{H}}\bs{Y}_0^{\sf{H}}=\bs{Y}_0\bs{Y}_0^{\sf{H}}.
\end{equation}
By differentiating (\ref{eq:tangentcurve}) we get
\begin{equation}\label{eq:tangentcurve2}
        \dot{\bs{Y}}(t)\bs{Y}(t)^{\sf{H}}+\bs{Y}(t)\dot{\bs{Y}}(t)^{\sf{H}}=\bs{0}.
\end{equation}
So we deduce that $\dot{\bs{Y}}(0)$ belongs to
\begin{equation}\label{eq:vertical1}
        \{\bs{Z}\in\mathbb{C}^{N\times r}: \bs{Z}\bs{Y}_0^{\sf{H}}+\bs{Y}_0\bs{Z}^{\sf{H}}=\bs{0}\},
\end{equation}
of which $\mathcal{T}_{\bs{Y}}\mathcal{M}$ is a subset. On the other side, let $F:\bs{Y}\mapsto \bs{Y}\bs{Y}^{\sf{H}}$, then (\ref{eq:vertical1}) is $\textrm{ker}(\textrm{D}F(\bs{Y_0}))$ and $F^{-1}(\bs{Y}_0\bs{Y}_0^{\sf{H}})=[\bs{Y}_0]$. Therefore, from \cite[Sec 3.5.7]{Absil_2009optimizationonManifolds} we know
\begin{equation}
        \mathcal{T}_{\bs{Y}}\mathcal{M}=\{\bs{Z}\in\mathbb{C}^{N\times r}: \bs{Z}\bs{Y}_0^{\sf{H}}+\bs{Y}_0\bs{Z}^{\sf{H}}=\bs{0}\}.
\end{equation}
Without loss of generality, we can set
\begin{equation}\label{eq:tangentcurve3}
        \dot{\bs{Y}}(t) = \bs{Y}(t)\bs{\Omega}(t),~~\bs{\Omega}(t)\in\mathbb{C}^{r\times r},
\end{equation}
since $\bs{Y}(t)\in\mathbb{C}_*^{N\times r}$ is full rank. Then we can replace equation (\ref{eq:tangentcurve3}) in (\ref{eq:tangentcurve2}) and obtain
\begin{equation}
        \bs{Y}(t)(\bs{\Omega}(t)+\bs{\Omega}(t)^{\sf{H}})\bs{Y}(t)^{\sf{H}}=\bs{0}.
\end{equation}
Therefore, the vertical space is given by
\begin{equation}
        \mathcal{V}_{\bs{Y}}=\{\bs{Y}_0\bs{\Omega}:\bs{\Omega}^{\sf{H}}=-\bs{\Omega}\}.
\end{equation}

\section{Proof of Proposition \ref{proposition:horizontal}}\label{append:horizontal}
The horizontal space is given by
\begin{equation}
     \mathcal{H}_{\bs{Y}} = \{\overline{\bs{\xi}}\in\mathcal{T}_{\bs{Y}}\overline{\mathcal{M}}:\overline{g}_{\bs{Y}}(\overline{\bs{\xi}},\overline{\bs{\zeta}})=0,~\forall \overline{\bs{\zeta}}\in\mathcal{V}_{\bs{Y}}\},
\end{equation}
that is
\begin{equation}
        \overline{g}_{\bs{Y}}(\overline{\bs{\xi}},\bs{Y}\bs{\Omega})=0, \forall \bs{\Omega}^{\sf{H}}=-\bs{\Omega}.
\end{equation}
Since 
\begin{align}
     \overline{g}_{\bs{Y}}(\overline{\bs{\xi}},\bs{Y}\bs{\Omega})&=\trace(\overline{\bs{\xi}}^{\sf{H}}\bs{Y}\bs{\Omega}+\bs{\Omega}^{\sf{H}}\bs{Y}^{\sf{H}}\overline{\bs{\xi}}) \nonumber\\
     &=\trace(\overline{\bs{\xi}}^{\sf{H}}\bs{Y}\bs{\Omega}-\bs{\Omega}\bs{Y}^{\sf{H}}\overline{\bs{\xi}}) \nonumber\\
     &=\trace((\overline{\bs{\xi}}^{\sf{H}}\bs{Y}-\bs{Y}^{\sf{H}}\overline{\bs{\xi}})\bs{\Omega}),
\end{align}
we know the horizontal space consists of all the elements $\overline{\bs{\xi}}$ that satisfies $\trace((\overline{\bs{\xi}}^{\sf{H}}\bs{Y}-\bs{Y}^{\sf{H}}\overline{\bs{\xi}})\bs{\Omega})=0$ for all $\bs{\Omega}^{\sf{H}}=-\bs{\Omega}$. Therefore, the horizontal space is
\begin{equation}
        \mathcal{H}_{\bs{Y}}=\{\bs{\xi}\in\mathbb{C}^{N\times r}:\bs{\xi}^{\sf{H}}\bs{Y}=\bs{Y}^{\sf{H}}\bs{\xi}\}.
\end{equation}
Suppose for a vecor $\bs{\xi}\in\mathcal{T}_{\bs{Y}}\mathcal{M}$ its projection onto the vertical space is given by $\bs{\xi}^{v}=\bs{Y}\bs{\Omega}_{\bs{\xi}}$, then the horizontal projection is given by $\bs{\xi}^{h}=\bs{\xi}-\bs{Y}\bs{\Omega}_{\bs{\xi}}$, and
\begin{equation}
        {\bs{\xi}^{h}}^{\sf{H}}\bs{Y}=\bs{Y}^{\sf{H}}{\bs{\xi}^{h}}.
\end{equation}
So we can find the $\bs{\Omega}_{\bs{\xi}}$ from
\begin{align}
        &(\bs{\xi}-\bs{Y}\bs{\Omega}_{\bs{\xi}})^{\sf{H}}\bs{Y}=\bs{Y}^{\sf{H}}(\bs{\xi}-\bs{Y}\bs{\Omega}_{\bs{\xi}}) \nonumber\\\Rightarrow ~& \bs{Y}^{\sf{H}}\bs{Y}\bs{\Omega}_{\bs{\xi}}+\bs{\Omega}_{\bs{\xi}}\bs{Y}^{\sf{H}}\bs{Y} = \bs{Y}^{\sf{H}}\bs{\xi}-\bs{\xi}^{\sf{H}}\bs{Y}.
\end{align}
Then we conclude that the horizontal projection of $\bs{\xi}$ is given by
\begin{equation}
        \Pi_{\bs{Y}}^{h}\bs{\xi}=\bs{\xi}-\bs{Y}\bs{\Omega},
\end{equation}
where $\bs{\Omega}$ is the solution to the Lyapunov equation
\begin{equation}
       \bs{Y}^{\sf{H}}\bs{Y}\bs{\Omega}+\bs{\Omega}\bs{Y}^{\sf{H}}\bs{Y} = \bs{Y}^{\sf{H}}\bs{\xi}-\bs{\xi}^{\sf{H}}\bs{Y}.
\end{equation}

\section{Computing the Riemannian Gradient and Hessian}\label{append:gradhess}
We first rewrite the objective function of (\ref{GLRM_BM}) as
\begin{equation}
       f(\bs{Y}) = \frac{1}{2}\sum_{i=1}^{l}|\langle \bs{B}_i,\bs{Y}\bs{Y}^{\sf{H}} \rangle-b_i|^2.
\end{equation}
The complex gradient of $f(\bs{Y})$ is given by
\begin{align}
        f^\prime(\bs{Y}) &= \sum_{i=1}^{l} (\langle \bs{B}_i,\bs{Y}\bs{Y}^{\sf{H}} \rangle-b_i)\bs{B}_i\bs{Y} +(\langle \bs{B}_i^{\sf{H}},\bs{Y}\bs{Y}^{\sf{H}} \rangle-b_i^*)\bs{B}_i^{\sf{H}}\bs{Y} \nonumber \\
        &=\sum_{i=1}^{l}(C_i\bs{B}_i+C_i^*\bs{B}_i^{\sf{H}})\bs{Y},
\end{align}
in which $C_i=\langle \bs{B}_i,\bs{Y}\bs{Y}^{\sf{H}} \rangle-b_i$. The Riemannian gradient $\overline{\textrm{grad}}f(\bs{Y})$ is derived from (\ref{eq:gradbar}), and we find that
\begin{align}
        \textrm{D}f(\bs{Y})[\bs{\xi}] &= \frac{1}{2}\sum_{i=1}^{l} \langle \bs{B}_i,\bs{\xi}\bs{Y}^{\sf{H}}+\bs{Y}\bs{\xi}^{\sf{H}} \rangle^*(\langle \bs{B}_i,\bs{Y}\bs{Y}^{\sf{H}}\rangle-b_i) \nonumber\\&\qquad+(\langle \bs{B}_i,\bs{Y}\bs{Y}^{\sf{H}}\rangle-b_i)^*\langle \bs{B}_i,\bs{\xi}\bs{Y}^{\sf{H}}+\bs{Y}\bs{\xi}^{\sf{H}} \rangle \nonumber\\
        &=\overline{g}_{\bs{Y}}((C_i\bs{B}_i+C_i^*\bs{B}_i^{\sf{H}})\bs{Y},\bs{\xi}).
\end{align}
Therefore, $\overline{\textrm{grad}}f(\bs{Y})=f^\prime(\bs{Y})$. Then we observe that $\overline{\textrm{grad}}f(\bs{Y})^{\sf{H}}\bs{Y}=\bs{Y}^{\sf{H}}\overline{\textrm{grad}}f(\bs{Y})$, i.e., $\overline{\textrm{grad}}f(\bs{Y})$ is already in the horizontal space $\mathcal{V}_{\bs{Y}}$. So the horizontal representation of Riemannian gradient is given by 
\begin{equation}
         \textrm{grad}f(\bs{Y})=\sum_{i=1}^{l}(C_i\bs{B}_i+C_i^*\bs{B}_i^{\sf{H}})\bs{Y}.
 \end{equation}

To derive the Riemannian Hessian (\ref{eq:rhess}), we compute
\begin{align}
        \textrm{D}\overline{\textrm{grad}}f(\bs{Y})[\bs{\eta}_{\bs{Y}}]=\sum_{i=1}^{l}(&{C_{\bs{\eta}}}_{i}\bs{B}_i\bs{Y} +C_i\bs{B}_i\bs{\eta}_Y  \nonumber\\&
        + {C_{\bs{\eta}}}_{i}^*\bs{B}_i^{\sf{H}}\bs{Y} +C_i^*\bs{B}_i^{\sf{H}}\bs{\eta}_\bs{Y}),
\end{align}
where ${C_{\bs{\eta}}}_{i}=\langle \bs{B}_i,\bs{Y}\bs{\eta}_\bs{Y}^{\sf{H}} +\bs{\eta}_\bs{Y}\bs{Y}^{\sf{H}}\rangle$.
We conclude that
\begin{align}
        \textrm{Hess}f(\bs{Y})[\bs{\eta}_{\bs{Y}}]=\Pi_{\bs{Y}}^{h}\biggl(&\sum_{i=1}^{l}({C_{\bs{\eta}}}_{i}\bs{B}_i\bs{Y}  +C_i\bs{B}_i\bs{\eta}_Y \nonumber\\&+ {C_{\bs{\eta}}}_{i}^*\bs{B}_i^{\sf{H}}\bs{Y} +C_i^*\bs{B}_i^{\sf{H}}\bs{\eta}_\bs{Y})\biggr).
\end{align}

\bibliographystyle{IEEEtran}
\bibliography{GMC}

\end{document}